\documentclass[
reprint,
superscriptaddress,
amsmath,amssymb,
aps,
prl,
floatfix
]{revtex4-2}

\usepackage{braket}
\usepackage{booktabs}
\usepackage{array}
\usepackage{bm}
\usepackage{hyperref}
\usepackage{amsthm}
\newtheorem{theorem}{Theorem}
\newtheorem{proposition}[theorem]{Proposition}

\theoremstyle{remark}

\begin{document}

\title{Embedding Stabilizer Codes and Leakage Correction in Multilevel Quantum Systems}

\author{Ali Abu-Nada}

\affiliation{Sharjah Maritime Academy, United Arab Emirates}

\author{Lian-Ao Wu}
\affiliation{Department of Physics, University of the Basque Country UPV/EHU,
48080 Bilbao, Spain and IKERBASQUE, Basque Foundation for Science, Bilbao 48011, Spain}
\affiliation{EHU Quantum Center, University of the Basque Country UPV/EHU,
Leioa, Biscay 48940, Spain}

\date{\today}

\begin{abstract}
Leakage beyond the computational subspace is a major source of error in
multilevel quantum hardware. We show that any \( [[n,k,d]] \) stabilizer
code can be embedded isometrically into a single \(D\)-dimensional
system while preserving its complete error-correcting structure. We
further derive a necessary and sufficient condition for exact leakage
correction, proving that leakage is correctable precisely when it does
not distinguish between logical states. These results establish a
unified framework for quantum error correction in multilevel quantum
systems.
\end{abstract}

\maketitle

Quantum error correction (QEC) provides the theoretical foundation for reliable quantum information processing by identifying when physical errors can be detected and exactly reversed through suitable recovery operations. Its mathematical cornerstone is the Knill--Laflamme (KL) theorem, which characterizes precisely when a family of errors is correctable by a recovery map \cite{KnillLaflamme1997}. Together with the stabilizer formalism \cite{Gottesman1997,Calderbank1998}, the KL conditions underlie essentially all modern approaches to fault-tolerant quantum computation.

Despite its success, the standard theory of QEC is fundamentally formulated for qubits. Stabilizer codes are defined on tensor-product Hilbert spaces of ideal two-level systems, whereas physical quantum hardware is intrinsically multilevel. Superconducting circuits \cite{Devoret2013,Koch2007}, trapped ions \cite{Cirac1995,Blatt2008}, neutral atoms \cite{Saffman2010,Browaeys2020}, semiconductor spin systems \cite{Loss1998,Russ2017}, and bosonic platforms \cite{Michael2016,Albert2018,Grimsmo2021} all possess additional degrees of freedom beyond the computational subspace used to encode logical qubits. These extra levels actively participate in the system dynamics and give rise to physical error mechanisms that have no counterpart in ideal qubit models.

Among these, leakage is particularly important. Leakage transfers population outside the computational manifold \cite{WuByrdLidar2002,Jing2015,Ghosh2013,Chen2016,McEwen2021}, thereby changing the effective Hilbert space occupied by the encoded state. Unlike conventional errors, which act within a fixed code space, leakage moves quantum information beyond the domain on which a stabilizer code is defined. Consequently, the standard qubit-based framework does not provide a complete description of error correction in realistic multilevel devices. While substantial progress has been made in the detection, mitigation, and suppression of leakage, and recent experiments have demonstrated effective leakage-removal strategies in QEC protocols \cite{Miao2023Leakage}, a fundamental theoretical question remains unresolved: Can the principles of stabilizer quantum error correction be extended to multilevel quantum systems, and if so, what is the exact criterion that determines when leakage is correctable?

In this Letter, we answer both questions. We first prove that every $[[n,k,d]]$ stabilizer code admits an exact isometric embedding into a single $D$-dimensional quantum system while preserving its complete error-correcting structure, including the code space, stabilizer algebra, logical operators, syndrome measurements, recovery maps, and Knill--Laflamme conditions. This establishes a rigorous equivalence between conventional qubit codes and their multilevel realizations and demonstrates that stabilizer QEC possesses a representation-independent structure extending beyond tensor-product qubit architectures.

We then address genuine physical leakage beyond the embedded computational manifold. By decomposing arbitrary physical error operators into embedded and leakage components, we derive a Leakage Knill--Laflamme theorem that provides necessary and sufficient conditions for exact leakage correction. The theorem shows that leakage is exactly correctable if and only if the leakage process acquires no information about the encoded logical state. Equivalently, leakage errors are correctable precisely when they remain indistinguishable on the code space. In this regime, leakage maps logical states into orthogonal syndrome sectors that can be identified and reversed by an explicit recovery operation.

Together, these results establish a unified theory of quantum error correction in multilevel quantum systems. They extend the stabilizer formalism beyond idealized qubit models and provide an exact characterization of correctable leakage. More broadly, they place leakage and conventional errors within a common mathematical framework, thereby providing a rigorous foundation for fault-tolerant quantum information processing in superconducting, atomic, spin-based, and bosonic quantum architectures.

\textit{Embedding qubit codes into multilevel hardware}---
QEC is conventionally formulated in an
$n$-qubit Hilbert space
$\mathcal H_{\mathrm q}=(\mathbb C^2)^{\otimes n}$,
whereas realistic quantum devices are intrinsically multilevel and are
described by a larger Hilbert space $\mathcal H_D$ with
$D\ge 2^n$. We decompose the physical Hilbert space as $
\mathcal H_D
=
\mathcal H_{\mathrm{emb}}
\oplus
\mathcal H_{\mathrm{leak}}$, where $\mathcal H_{\mathrm{emb}}$ is a $2^n$-dimensional computational
manifold and $\mathcal H_{\mathrm{leak}}$ contains the remaining
physical levels. This decomposition distinguishes the computational manifold that carries the encoded quantum information from the additional physical levels that become populated during leakage processes.

To establish an exact correspondence between the qubit
description and the multilevel physical system, we define
the embedding map
$\Phi:\mathcal H_{\mathrm q}\rightarrow\mathcal H_{\mathrm{emb}}$ by

\begin{equation}
\Phi
=
\sum_{m=0}^{2^n-1}
|m\rangle\langle b(m)|,
\qquad
\Phi^\dagger
=
\sum_{m=0}^{2^n-1}
|b(m)\rangle\langle m|,
\label{eq:Phi}
\end{equation}

where $|b(m)\rangle$ denotes the computational basis state associated
with the binary representation of the integer $m$. Thus, $
\Phi |b(m)\rangle = |m\rangle$ establishing a one-to-one correspondence between the qubit basis and
the computational manifold of the multilevel device.
The embedding map satisfies
\begin{equation}
\Phi^\dagger\Phi
=
I_{\mathcal H_{\mathrm q}},
\qquad
\Phi\Phi^\dagger
=
P_{\mathrm{emb}},
\label{eq:isometry}
\end{equation}
where \(I_{\mathcal H_{\mathrm q}}\) denotes the identity operator on the
\(n\)-qubit Hilbert space
\(\mathcal H_{\mathrm q}=(\mathbb C^2)^{\otimes n}\), and
\(P_{\mathrm{emb}}\) is the projector onto the embedded subspace
\(\mathcal H_{\mathrm{emb}}\subseteq\mathcal H_D\).
Equation~(\ref{eq:isometry}) shows that \(\Phi\) is an isometric
embedding of \(\mathcal H_{\mathrm q}\) into the multilevel system,
preserving all inner products and therefore all encoded quantum
information. Consequently, logical states, stabilizer generators,
syndrome measurements, and recovery operations admit equivalent
representations within the multilevel device. The proof is given in
Supplemental Material Sec.~S2.

\textit{Preservation of quantum error correction}--- Having established the embedding map, the next fundamental question is
whether the embedded multilevel realization retains the error-correcting
capability of the original stabilizer code. Since exact QEC  is completely characterized by the KL conditions, it is sufficient to determine whether these conditions are
preserved under the embedding. If they are, then the embedded code is
not merely a representation of the original code but an exact
error-correcting equivalent.

For a stabilizer code
$\mathcal C\subseteq\mathcal H_{\mathrm q}$,
the KL conditions are

\begin{equation}
P_{\mathcal C}
E_a^\dagger
E_b
P_{\mathcal C}
=
\alpha_{ab}
P_{\mathcal C},
\label{eq:KL}
\end{equation}

where $P_{\mathcal C}$ denotes the projector onto the logical code
space, $\{E_a\}$ is the set of correctable error operators, and
$\alpha=[\alpha_{ab}]$ is a Hermitian positive-semidefinite matrix.
These conditions are both necessary and sufficient for the existence of
an exact recovery operation.

Under the embedding, the logical code projector and error operators are
mapped according to

\begin{equation}
\widetilde P_{\mathcal C}
=
\Phi P_{\mathcal C}\Phi^\dagger,
\qquad
\widetilde E_a
=
\Phi E_a\Phi^\dagger.
\end{equation}

The preservation of the KL conditions follows from the isometric nature
of the embedding. As shown in Sec.~S2 of the Supplemental Material, the
embedding map satisfies $\Phi^\dagger\Phi = I_{\mathcal H_q}$.  Consequently, this map  preserves inner products, operator
products, and algebraic relations. In particular, for any operators
$A$ and $B$ acting on the original qubit Hilbert space, $(\Phi A\Phi^\dagger)
(\Phi B\Phi^\dagger) = \Phi AB\Phi^\dagger$, which implies that all operator identities satisfied in the original
qubit description are transferred exactly to the embedded multilevel
description. Therefore, the complete structure underlying QEC is preserved.

The following theorem formalizes this result.

\begin{theorem}\label{thrm1}
Let
$\Phi:\mathcal H_{\mathrm q}\rightarrow\mathcal H_D$
be the isometric embedding map defined above. If the original stabilizer
code satisfies the Knill--Laflamme conditions (\ref{eq:KL}), then the
embedded code
$\widetilde{\mathcal C}=\Phi\mathcal C$
satisfies

\begin{equation}
\widetilde P_{\mathcal C}
\widetilde E_a^\dagger
\widetilde E_b
\widetilde P_{\mathcal C}
=
\alpha_{ab}
\widetilde P_{\mathcal C}.
\label{eq:KL_embedded}
\end{equation}

Hence, the embedded code possesses exactly the same error-correcting
capability as the original stabilizer code.
\end{theorem}

\begin{proof}
Using
$\widetilde P_{\mathcal C}
=
\Phi P_{\mathcal C}\Phi^\dagger$,
$\widetilde E_a
=
\Phi E_a\Phi^\dagger$,
and
$\Phi^\dagger\Phi
=
I_{\mathcal H_q}$,

\begin{align}
\widetilde P_{\mathcal C}
\widetilde E_a^\dagger
\widetilde E_b
\widetilde P_{\mathcal C}
&=
\Phi
P_{\mathcal C}
E_a^\dagger
E_b
P_{\mathcal C}
\Phi^\dagger
\nonumber\\
&=
\alpha_{ab}
\Phi P_{\mathcal C}\Phi^\dagger
\nonumber\\
&=
\alpha_{ab}
\widetilde P_{\mathcal C},
\end{align}

where the second line follows directly from the original
KL  conditions.
\end{proof}

 Theorem~\ref{thrm1} shows that the embedding map preserves the complete
error-correcting structure of the original code. Consequently, every
$[[n,k,d]]$ stabilizer code admits an equivalent realization within a
single $D$-dimensional quantum system. Stabilizer generators, logical
operators, syndrome projectors, syndrome measurements, and recovery
maps are transferred according to $\widetilde A = \Phi A\Phi^\dagger$, while preserving all algebraic relations. Detailed derivations of the
stabilizer algebra, logical operators, syndrome projectors, syndrome
measurements, and recovery operations are provided in Sec.~S3 of the
Supplemental Material.

We illustrate
the embedding framework using representative stabilizer codes,
including repetition, nondegenerate, and CSS codes. In each case, the
embedding map transfers the code space, stabilizer generators, logical
operators, syndrome structure, and recovery operations to the multilevel
system exactly as predicted by the general theory. These examples
demonstrate the broad applicability of the framework across standard
classes of stabilizer QEC  codes.

\textit{Example: Three-Qubit Repetition Code.}---
As a simple illustration, consider the three-qubit repetition code with
logical states
\(
|0_L\rangle=|000\rangle
\)
and
\(
|1_L\rangle=|111\rangle,
\)
stabilizer generators
\(
S_1=Z_1Z_2
\)
and
\(
S_2=Z_2Z_3,
\)
and correctable error set
\(
\mathcal E_3=\{I,X_1,X_2,X_3\}.
\)

The embedding map \(\Phi\) identifies the eight computational basis
states with the levels of a single \(D=8\) system. In particular, \(
|000\rangle \leftrightarrow |0\rangle\) and \(
|111\rangle \leftrightarrow |7\rangle\), so that the embedded logical code is \(\widetilde{\mathcal C}
=
\mathrm{span}\{|0\rangle,|7\rangle\}
\). The stabilizers are transferred according to \(
\widetilde S_r=\Phi S_r\Phi^\dagger\),
yielding \(\widetilde S_1
= \operatorname{diag}(1,1,-1,-1,-1,-1,1,1)\) and \(\widetilde S_2
=\operatorname{diag}(1,-1,-1,1,1,-1,-1,1)\).

Similarly, the embedded bit-flip errors are

\begin{equation}
\widetilde X_j
=
\Phi X_j\Phi^\dagger
=
\sum_{m=0}^{7}
|m\oplus2^{\,3-j}\rangle\langle m|,
\qquad
j=1,2,3,
\end{equation}

where \(\oplus\) denotes bitwise exclusive-OR.

Because the embedding is isometric, the syndrome structure, recovery
operations, and Knill--Laflamme conditions are preserved exactly.
Hence the embedded eight-level realization reproduces the complete
error-correcting behavior of the original repetition code. Detailed
syndrome tables and recovery constructions are presented in the
Supplemental Material.

The framework extends directly to the [[5,1,3]] perfect code and to the complete CSS structure of the [[7,1,3]] Steane code. Explicit constructions of the embedded logical states, stabilizer generators, syndrome tables, and corresponding recovery operators are presented in the Supplemental Material.

The three examples above demonstrate the preservation of stabilizer
codes under error processes that remain confined to the embedded
computational manifold. However, multilevel quantum systems possess
additional physical levels beyond \(\mathcal H_{\mathrm{emb}}\), and
physical noise may couple the encoded states to these levels. Such
processes correspond to genuine leakage errors and represent a distinct
error mechanism that is absent in the original qubit description.
Because leakage drives the system outside the embedded code space, its
correctability is not determined by the standard stabilizer
Knill--Laflamme conditions alone. Instead, leakage correction requires
a separate set of conditions, which we formulate below in terms of the
physical leakage operators connecting
\(\mathcal H_{\mathrm{emb}}\) and
\(\mathcal H_{\mathrm{leak}}\).

\textit{Correcting physical leakage beyond the embedded manifold.}---
The embedding theorem establishes that every logical codeword,
syndrome subspace, correctable error sector, and recovery operation of
the original stabilizer code is faithfully transferred to the embedded
manifold \(\mathcal H_{\rm emb}\). Consequently, all errors acting
entirely within \(\mathcal H_{\rm emb}\) are governed by the standard
KL framework.

Real multilevel quantum devices, however, possess additional physical
levels outside \(\mathcal H_{\rm emb}\). Imperfect control,
system--environment interactions, and weak anharmonicity may populate
these levels
\cite{Preskill2018,Fowler2012,Koch2007,McEwen2021,Miao2023Leakage},
producing leakage processes that are absent from the original qubit
description. Unlike ordinary correctable errors, which move states
between syndrome sectors within \(\mathcal H_{\rm emb}\), leakage errors
transfer population out of the embedded manifold altogether. The central
question is therefore whether such excursions merely change the physical
support of the encoded state or reveal information about the underlying
logical state.

To address this question, we decompose the physical Hilbert space as $
\mathcal H_D
=
\mathcal H_{\rm emb}
\oplus
\mathcal H_{\rm leak}$, where \(\mathcal H_{\rm leak}\) denotes the orthogonal complement of the embedded manifold and $
P_{\rm leak}
=
I_D-P_{\rm emb}$ is the corresponding projector. Relative to this decomposition, an
arbitrary physical error admits the block representation

\begin{equation}
F_a=
\begin{pmatrix}
A_a & B_a\\
L_a & C_a
\end{pmatrix},
\label{eq:S_block_error}
\end{equation}

with respect to the ordered basis
\(
\mathcal H_{\rm emb}\oplus\mathcal H_{\rm leak}.
\)
Here

\begin{equation}
A_a
=
P_{\rm emb}F_aP_{\rm emb},
\qquad
L_a
=
P_{\rm leak}F_aP_{\rm emb},
\label{eq:S_emb_leak_blocks}
\end{equation}

describe, respectively, evolution confined to the embedded manifold and
transitions from the embedded manifold into the leakage sector. The
remaining blocks act on states already supported in
\(\mathcal H_{\rm leak}\) and therefore play no role in the evolution
of an initially encoded state.

Let $
\widetilde{\mathcal C}
=
\Phi\mathcal C$ and $
P
\equiv
\widetilde P_{\mathcal C}
=
\Phi P_{\mathcal C}\Phi^\dagger $ denote the embedded code space and its projector, respectively. Since
\(
\widetilde{\mathcal C}
\subseteq
\mathcal H_{\rm emb},
\)
it follows that $
P_{\rm emb}P=P$ and $P_{\rm leak}P=0$. Consequently, $
F_aP
=
F_aP_{\rm emb}P$. Inserting the resolution of the identity $
I_D=P_{\rm emb}+P_{\rm leak}$ on the output side yields

\begin{align}
F_aP
&=
(P_{\rm emb}+P_{\rm leak})
F_aP_{\rm emb}P \nonumber \\
&=
(A_a+L_a)P.
\label{eq:S_error_decomposition_on_code}
\end{align}

Equation~(\ref{eq:S_error_decomposition_on_code}) shows that the action
of any physical error on an encoded state naturally separates into two
distinct contributions: an embedded component \(A_aP\), which remains
within \(\mathcal H_{\rm emb}\) and is governed by the standard
embedded error-correction framework, and a leakage component
\(L_aP\), which transfers population into
\(\mathcal H_{\rm leak}\).

Because $A_aP \subseteq \mathcal H_{\rm emb}$ and $L_aP \subseteq \mathcal H_{\rm leak}$, and the subspaces
\(\mathcal H_{\rm emb}\) and \(\mathcal H_{\rm leak}\) are orthogonal,
the mixed overlap terms vanish identically,

\begin{equation}
PA_a^\dagger L_bP
=
PL_a^\dagger A_bP
=
0.
\label{eq:S_cross_terms}
\end{equation}

The KL overlap therefore decomposes as

\begin{equation}
PF_a^\dagger F_bP
=
PA_a^\dagger A_bP
+
PL_a^\dagger L_bP.
\label{eq:S_KL_split}
\end{equation}

The first contribution is determined entirely by the embedded
stabilizer code and corresponds to ordinary correctable errors within
\(\mathcal H_{\rm emb}\). The second contribution arises exclusively
from transitions into the leakage sector and represents the additional
constraint imposed by physical leakage. The theorem below shows that
exact correction remains possible precisely when this leakage term is
unable to distinguish between logical states.

The relevance of this result extends directly to contemporary
multilevel quantum platforms. In superconducting transmon devices,
for example, weak anharmonicity can induce transitions from the
computational states into higher excited levels, while in bosonic
encodings leakage may populate Fock states outside the designated
code manifold. As shown in the Supplemental Material, the theorem
provides an explicit syndrome construction and recovery procedure for
both embedded stabilizer codes and bosonic encodings.

\begin{theorem}[Exact correction of physical leakage]
\label{thm:main_physical_leakage}

Let \(P=\widetilde P_{\mathcal C}\) be the projector onto the embedded
code space
\(\widetilde{\mathcal C}\subseteq\mathcal H_{\rm emb}\). For each
physical error \(F_a\), define its in-manifold and outward-leakage
components by \(
A_a=P_{\rm emb}F_aP_{\rm emb}\) and \(
L_a=P_{\rm leak}F_aP_{\rm emb}\).
Assume that the in-manifold components satisfy \(
PA_a^\dagger A_bP=\alpha_{ab}P \) for all \(a,b\). Then a recovery operation exists that restores every
state initially encoded in \(\widetilde{\mathcal C}\) after the action
of the error family \(\{F_a\}\) if and only if
\[
PL_a^\dagger L_bP=\beta_{ab}P
\]
for some matrix \(\beta=[\beta_{ab}]\). In this case, the complete
Knill--Laflamme matrix is
\(\gamma=\alpha+\beta\), and \(\beta\) is Hermitian and positive
semidefinite.
\end{theorem}

\begin{proof}
We first establish the decomposition on which both directions of the
proof rely. Since
\(\widetilde{\mathcal C}\subseteq\mathcal H_{\rm emb}\), its projector
satisfies \(
P_{\rm emb}P=P\) and \(P_{\rm leak}P=0\).
Therefore \(F_aP
=
F_aP_{\rm emb}P\). Inserting
\(I_D=P_{\rm emb}+P_{\rm leak}\)
on the output side gives \(
F_aP
=
(P_{\rm emb}+P_{\rm leak})F_aP_{\rm emb}P
=
(A_a+L_a)P\). The mixed terms vanish because the two components have orthogonal
ranges. More explicitly,
\[
A_a^\dagger L_b
=
P_{\rm emb}F_a^\dagger
P_{\rm emb}P_{\rm leak}
F_bP_{\rm emb}
=
0,
\]
since \(P_{\rm emb}P_{\rm leak}=0\). Similarly,
\[
L_a^\dagger A_b
=
P_{\rm emb}F_a^\dagger
P_{\rm leak}P_{\rm emb}
F_bP_{\rm emb}
=
0.
\]
Consequently,
\begin{align}
PF_a^\dagger F_bP
&=
P(A_a^\dagger+L_a^\dagger)
(A_b+L_b)P
\nonumber\\
&=
PA_a^\dagger A_bP
+
PL_a^\dagger L_bP.
\label{eq:main_physical_KL_split}
\end{align}

We now prove necessity. Suppose that the action of the complete physical
error family \(\{F_a\}\) on states encoded in
\(\widetilde{\mathcal C}\) is exactly correctable. By the
Knill--Laflamme theorem, there exists a matrix
\(\gamma=[\gamma_{ab}]\) such that
\[
PF_a^\dagger F_bP
=
\gamma_{ab}P
\]
for all \(a,b\). Combining this condition with
Eq.~\eqref{eq:main_physical_KL_split} and the assumed in-manifold
relation \(PA_a^\dagger A_bP=\alpha_{ab}P\), we obtain
\begin{align}
PL_a^\dagger L_bP
&=
PF_a^\dagger F_bP
-
PA_a^\dagger A_bP
\nonumber\\
&=
(\gamma_{ab}-\alpha_{ab})P.
\end{align}
Thus the leakage components necessarily satisfy
\[
PL_a^\dagger L_bP
=
\beta_{ab}P,
\qquad
\beta_{ab}
=
\gamma_{ab}-\alpha_{ab}.
\]
This proves necessity.

We next prove sufficiency. Suppose that
\[
PL_a^\dagger L_bP
=
\beta_{ab}P
\]
for all \(a,b\). Substituting this condition and
\(PA_a^\dagger A_bP=\alpha_{ab}P\) into
Eq.~\eqref{eq:main_physical_KL_split} gives
\begin{align}
PF_a^\dagger F_bP
&=
\alpha_{ab}P+\beta_{ab}P
\nonumber\\
&=
(\alpha_{ab}+\beta_{ab})P.
\end{align}
Hence the complete physical error family satisfies the
Knill--Laflamme conditions with \(
\gamma_{ab}
=
\alpha_{ab}+\beta_{ab}\). The Knill--Laflamme theorem therefore guarantees the existence of a
recovery operation that exactly restores every state initially encoded
in \(\widetilde{\mathcal C}\). This proves sufficiency.

It remains to establish the structure of \(\beta\). Taking the adjoint
of
\(PL_a^\dagger L_bP=\beta_{ab}P\)
gives
\[
PL_b^\dagger L_aP
=
\beta_{ab}^{*}P.
\]
Interchanging \(a\) and \(b\) in the original condition gives
\[
PL_b^\dagger L_aP
=
\beta_{ba}P.
\]
Because \(P\neq0\), the two expressions imply \(
\beta_{ba}
=
\beta_{ab}^{*}\), so \(\beta\) is Hermitian.

To prove positive semidefiniteness, let
\(\ket{\psi}\in\widetilde{\mathcal C}\) be any normalized encoded state,
so that \(P\ket{\psi}=\ket{\psi}\), and let
\(\boldsymbol c=(c_1,c_2,\ldots)^T\) be an arbitrary complex vector.
Then
\begin{align}
\boldsymbol c^\dagger\beta\boldsymbol c
&=
\sum_{a,b}c_a^{*}\beta_{ab}c_b
\nonumber\\
&=
\sum_{a,b}c_a^{*}
\bra{\psi}L_a^\dagger L_b\ket{\psi}
c_b
\nonumber\\
&=
\bra{\psi}
\left(\sum_a c_aL_a\right)^\dagger
\left(\sum_b c_bL_b\right)
\ket{\psi}
\nonumber\\
&=
\left\|
\sum_a c_aL_a\ket{\psi}
\right\|^2
\geq0.
\end{align}
Therefore \(\beta\) is positive semidefinite. Since it is also
Hermitian, the spectral theorem guarantees a unitary matrix \(U\) such
that
\[
U^\dagger\beta U
=
\operatorname{diag}(\lambda_\mu),
\qquad
\lambda_\mu\geq0.
\]
\end{proof}
Theorem~\ref{thm:main_physical_leakage}  identifies the precise obstruction introduced by physical
leakage. In an orthonormal logical basis
\(\{\ket{\widetilde i_L}\}\), the condition becomes \(
\bra{\widetilde i_L}L_a^\dagger L_b\ket{\widetilde j_L}
=
\beta_{ab}\delta_{ij}
\). Thus the leakage overlaps are independent of the encoded logical state:
the leakage process may reveal which error occurred, but it carries no
information about the protected logical amplitudes. Population outside
the computational manifold is therefore not, by itself, an
irreversible loss of quantum information. Leakage becomes
uncorrectable only when the leaked states distinguish the logical
degrees of freedom.

Diagonalizing \(\beta\) produces independent leakage modes whose images
of the code form mutually orthogonal syndrome sectors. Their explicit
construction, the associated syndrome projectors and measurements, and
a completely positive trace-preserving recovery channel are given in
the Supplemental Material, together with examples for both embedded
stabilizer codes and bosonic encodings.

Theorem~\ref{thm:main_physical_leakage} also motivates a leakage-correction break-even point. Recent experiments have demonstrated break-even quantum error correction by extending logical lifetimes beyond those of the best physical components \cite{Sivak2023BreakEven}. The present theorem provides the theoretical conditions under which analogous leakage-correction break-even behavior may be achieved in multilevel quantum hardware.

\textit{Conclusion}---We established a unified framework for quantum
error correction in multilevel systems. Any \( [[n,k,d]] \) stabilizer
code can be embedded isometrically into a single \(D\)-dimensional
system while preserving its code space, operator algebra, syndrome
structure, recovery operations, and Knill--Laflamme conditions.

We further derived a necessary and sufficient condition for correcting
genuine physical leakage beyond the embedded manifold. When the
leakage overlaps are independent of the logical state, the leaked
states form orthogonal syndrome sectors and can be identified and
recovered without revealing the encoded information. The framework
therefore extends stabilizer quantum error correction from idealized
qubit models to intrinsically multilevel quantum hardware and provides
a basis for designing hardware-adapted leakage-correction protocols.
\bibliographystyle{apsrev4-2}
\bibliography{main}
\clearpage

\onecolumngrid

\begin{center}
{\Large\bfseries Supplemental Material}\\[0.4cm]
{\bfseries Embedding Stabilizer Codes and Leakage Correction in Multilevel Quantum Systems}\\[0.25cm]
Ali Abu-Nada and Lian-Ao Wu
\end{center}

\vspace{0.5cm}

\section{S1. Embedding Map and Binary Correspondence}
\label{s1}

The embedding map introduced in the main text establishes an explicit
correspondence between an $n$-qubit Hilbert space and the computational
manifold of a multilevel quantum system. Physically, the construction
identifies the first $2^n$ levels of the multilevel device with the
$2^n$ computational basis states of the qubit register, thereby
allowing quantum information encoded in qubits to be represented
exactly within a multilevel quantum system.

Let $ \mathcal H_{\mathrm q} = (\mathbb C^2)^{\otimes n}$ denote the $n$-qubit Hilbert space. Its computational basis consists of all binary strings of length $n$, $
|x_0x_1\cdots x_{n-1}\rangle$, where,  $x_j\in\{0,1\}$. Since there are $2^n$ such basis states, each binary string may be
associated with a unique integer between $0$ and $2^n-1$. Specifically,
the binary string
$x_0x_1\cdots x_{n-1}$ corresponds to the integer

\begin{equation}
m
=
\sum_{j=0}^{n-1}
2^{\,n-1-j}x_j.
\label{eq:binaryExpansionSM}
\end{equation}

This is simply the standard conversion from binary notation to decimal
notation. For example,

\[
000 \rightarrow 0,
\qquad
001 \rightarrow 1,
\qquad
011 \rightarrow 3,
\qquad
101 \rightarrow 5,
\qquad
111 \rightarrow 7.
\]

For notational convenience, we denote by $
|b(m)\rangle$ the computational basis state whose binary digits represent the integer
$m$. Thus,

\[
|b(0)\rangle = |000\cdots0\rangle,
\]

\[
|b(1)\rangle = |000\cdots1\rangle,
\]

and, for $n=3$,

\[
|b(5)\rangle = |101\rangle,
\qquad
|b(7)\rangle = |111\rangle.
\]

More generally, $|b(m)\rangle$ denotes the unique computational basis
state whose binary representation corresponds to the integer $m$.

The embedding is defined by identifying each computational basis state
with the multilevel basis state carrying the same integer label,

\begin{equation}
|b(m)\rangle
\longmapsto
|m\rangle,
\qquad
m=0,\ldots,2^n-1.
\label{eq:mappingRuleSM}
\end{equation}

In other words, the computational basis state corresponding to the
binary representation of $m$ is mapped to the physical level labelled
by the same integer $m$.

By linearity, this correspondence uniquely defines the embedding map

\begin{equation}
\Phi
=
\sum_{m=0}^{2^n-1}
|m\rangle
\langle b(m)|,
\qquad
\Phi^\dagger
=
\sum_{m=0}^{2^n-1}
|b(m)\rangle
\langle m|.
\label{eq:PhiCompactSM}
\end{equation}

The action of $\Phi$ on the computational basis is therefore

\begin{equation}
\Phi |b(m)\rangle
=
|m\rangle,
\qquad
m=0,\ldots,2^n-1.
\label{eq:PhiActionSM}
\end{equation}

Consequently, every computational basis state of the qubit Hilbert
space is assigned to a unique level of the multilevel system. The map
$\Phi$ therefore establishes a one-to-one and information-preserving
correspondence between the qubit Hilbert space and the computational
manifold of the multilevel device.

As an illustration, consider the case $n=3$. The computational basis
states

\[
|000\rangle,
|001\rangle,
|010\rangle,
|011\rangle,
|100\rangle,
|101\rangle,
|110\rangle,
|111\rangle
\]

correspond respectively to the integers

\[
0,1,2,3,4,5,6,7.
\]

The embedding map becomes

\begin{equation}
\begin{aligned}
\Phi
={}&
|0\rangle\langle000|
+|1\rangle\langle001|
+|2\rangle\langle010|
+|3\rangle\langle011|\\
&
+|4\rangle\langle100|
+|5\rangle\langle101|
+|6\rangle\langle110|
+|7\rangle\langle111|.
\end{aligned}
\label{eq:PhiThreeQubitSM}
\end{equation}

For example,

\begin{equation}
\Phi |101\rangle
=
|5\rangle,
\end{equation}

because the binary string $101$ represents the integer $5$. Similarly,

\[
\Phi |000\rangle = |0\rangle,
\qquad
\Phi |011\rangle = |3\rangle,
\qquad
\Phi |111\rangle = |7\rangle.
\]

The complete correspondence for the three-qubit case is summarized in
Table~\ref{tab:embedding}.

\begin{table}[h]
\centering
\caption{Embedding of the three-qubit computational basis into an
eight-level quantum system.}
\label{tab:embedding}
\begin{tabular}{ccc}
\hline
Integer $m$ & Computational basis state $|b(m)\rangle$ & Embedded state \\
\hline
0 & $|000\rangle$ & $|0\rangle$ \\
1 & $|001\rangle$ & $|1\rangle$ \\
2 & $|010\rangle$ & $|2\rangle$ \\
3 & $|011\rangle$ & $|3\rangle$ \\
4 & $|100\rangle$ & $|4\rangle$ \\
5 & $|101\rangle$ & $|5\rangle$ \\
6 & $|110\rangle$ & $|6\rangle$ \\
7 & $|111\rangle$ & $|7\rangle$ \\
\hline
\end{tabular}
\end{table}

The construction extends directly to arbitrary $n$, providing a
canonical embedding of the $2^n$-dimensional qubit Hilbert space into
the computational manifold of a $D$-dimensional quantum system.

\section{S2. Proof of the Isometric Properties of the Embedding Map}
\label{s2}

For the embedding map construction to be physically meaningful, it must
preserve all quantum information stored in the original qubit register.
In particular, embedding a qubit state into the multilevel system and
subsequently decoding it should reproduce the original state exactly.
Mathematically, this requirement is expressed by the relations

\begin{equation}
\Phi^\dagger\Phi
=
I_{\mathcal H_q},
\qquad
\Phi\Phi^\dagger
=
P_{\mathrm{emb}},
\label{eq:isometrySM}
\end{equation}

which show that $\Phi$ is an isometric embedding from the qubit Hilbert
space into the computational manifold of the multilevel device.

The first relation,
$\Phi^\dagger\Phi=I_{\mathcal H_q}$,
states that embedding followed by decoding leaves every qubit state
unchanged. Consequently, the embedding process introduces no loss of
information and preserves all inner products, amplitudes, and quantum
coherences.

The second relation,
$\Phi\Phi^\dagger=P_{\mathrm{emb}}$,
characterizes the image of the embedding. It shows that the embedded
states occupy precisely the computational manifold
$\mathcal H_{\mathrm{emb}}$, while any population outside this subspace
belongs to the leakage sector. Thus, $P_{\mathrm{emb}}$ acts as the
projector onto the physically relevant computational manifold.

The embedding operator is

\begin{equation}
\Phi
=
\sum_{m=0}^{2^n-1}
|m\rangle
\langle b(m)|,
\label{eq:SM_Phi}
\end{equation}

where $|b(m)\rangle$ denotes the computational basis state whose binary
representation corresponds to the integer $m$. Its adjoint is

\begin{equation}
\Phi^\dagger
=
\sum_{m=0}^{2^n-1}
|b(m)\rangle
\langle m|.
\label{eq:SM_PhiDag}
\end{equation}

Since the computational basis
$\{|b(m)\rangle\}$ of $\mathcal H_q$
and the multilevel basis
$\{|m\rangle\}$ of $\mathcal H_{\mathrm{emb}}$
are both orthonormal, we have

\begin{equation}
\langle b(m)|b(k)\rangle
=
\delta_{mk},
\qquad
\langle m|k\rangle
=
\delta_{mk}.
\end{equation}

Using these relations,

\begin{align}
\Phi^\dagger\Phi
&=
\left(
\sum_m |b(m)\rangle\langle m|
\right)
\left(
\sum_k |k\rangle\langle b(k)|
\right)
\nonumber\\
&=
\sum_{m,k}
|b(m)\rangle
\langle m|k\rangle
\langle b(k)|
\nonumber\\
&=
\sum_{m,k}
\delta_{mk}
|b(m)\rangle
\langle b(k)|
\nonumber\\
&=
\sum_m
|b(m)\rangle
\langle b(m)|
\nonumber\\
&=
I_{\mathcal H_q}.
\label{eq:SM_isometry1}
\end{align}

This result confirms that every qubit state is recovered exactly after
embedding and decoding. Explicitly, for any
$|\psi\rangle\in\mathcal H_q$, $
\Phi^\dagger\Phi|\psi\rangle
=
|\psi\rangle$. Thus, the embedding preserves the complete quantum information carried
by the qubit register. We next evaluate

\begin{align}
\Phi\Phi^\dagger
&=
\left(
\sum_m |m\rangle\langle b(m)|
\right)
\left(
\sum_k |b(k)\rangle\langle k|
\right)
\nonumber\\
&=
\sum_{m,k}
|m\rangle
\langle b(m)|b(k)\rangle
\langle k|
\nonumber\\
&=
\sum_{m,k}
\delta_{mk}
|m\rangle
\langle k|
\nonumber\\
&=
\sum_m
|m\rangle
\langle m|.
\label{eq:SM_isometry2}
\end{align}

Defining $
P_{\mathrm{emb}}
=
\sum_{m=0}^{2^n-1}
|m\rangle\langle m|$, we obtain $
\Phi\Phi^\dagger
=
P_{\mathrm{emb}}$
.

The operator $P_{\mathrm{emb}}$ is the orthogonal projector onto the
embedded computational manifold. Physically, it removes any component
outside the computational manifold and retains only the embedded qubit
degrees of freedom. Therefore, the image of $\Phi$ is exactly
$\mathcal H_{\mathrm{emb}}$.

As a simple illustration, consider a multilevel state $
|\Psi\rangle
=
\alpha |1\rangle
+
\beta |8\rangle$, where $|1\rangle$ belongs to the computational manifold and
$|8\rangle$ belongs to the leakage sector. Applying the projector gives $
P_{\mathrm{emb}}|\Psi\rangle
=
\alpha |1\rangle$, showing that the leakage component is removed while the computational
component is preserved.

\section{S3. Preservation of the Logical and Stabilizer-Code Structure}
\label{S3}

The KL conditions established in the main text guarantee
that the embedded code retains the error-correcting capability of the
original stabilizer code. A complete equivalence, however, requires
more than preservation of the KL conditions alone. One
must also show that the embedding preserves the geometry of the logical
code space, the operator algebra, stabilizer generators, logical
operators, syndrome assignments, syndrome measurements, syndrome
projectors, and recovery operations.

In this section, we establish each of these properties directly from
the isometric relations $i.e.$$\Phi^\dagger\Phi
=
I_{\mathcal H_{\rm q}}$ and $\Phi\Phi^\dagger
=
P_{\rm emb}$. Throughout this section, the embedded image of an operator
\(A\in\mathcal L(\mathcal H_{\rm q})\) is denoted by $
\widetilde A
=
\Phi A\Phi^\dagger$. The operator \(\widetilde A\) acts nontrivially only within
\(\mathcal H_{\rm emb}\). In particular, $
P_{\rm emb}\widetilde A
=
\widetilde A P_{\rm emb}
=
\widetilde A$. Thus, all algebraic equivalences stated below are exact on the embedded
computational manifold.

\subsection{S3.1 Equivalence of the Original and Embedded Code Spaces}
\label{S3.1}

Let \(\mathcal C\subseteq\mathcal H_{\rm q}\) be the code space of an
arbitrary \([[n,k,d]]\) stabilizer code. Its realization in the
multilevel system is defined as the image of \(\mathcal C\) under the
embedding map,

\begin{equation}
\widetilde{\mathcal C}
=
\Phi\mathcal C
=
\left\{
\Phi|\psi_L\rangle:
|\psi_L\rangle\in\mathcal C
\right\}
\subseteq
\mathcal H_{\rm emb}.
\label{eq:S3_embedded_code}
\end{equation}

Physically, Eq.~(\ref{eq:S3_embedded_code}) assigns a multilevel state
\(|\widetilde\psi_L\rangle=\Phi|\psi_L\rangle\) to every logical state
\(|\psi_L\rangle\in\mathcal C\). We now prove that this assignment
preserves the encoded information and that each embedded codeword
originates from exactly one logical state.

First, \(\widetilde{\mathcal C}\) is a linear subspace. Indeed, for
\(|\widetilde\psi_L\rangle=\Phi|\psi_L\rangle\) and
\(|\widetilde\varphi_L\rangle=\Phi|\varphi_L\rangle\), with
\(|\psi_L\rangle,|\varphi_L\rangle\in\mathcal C\), linearity of
\(\Phi\) gives

\begin{equation}
a|\widetilde\psi_L\rangle+b|\widetilde\varphi_L\rangle
=
\Phi\!\left(
a|\psi_L\rangle+b|\varphi_L\rangle
\right)
\in\widetilde{\mathcal C}
\label{eq:S3_subspace}
\end{equation}

for all \(a,b\in\mathbb C\), because
\(a|\psi_L\rangle+b|\varphi_L\rangle\in\mathcal C\).

The preservation of logical information follows from the isometry
relation

\begin{equation}
\Phi^\dagger\Phi
=
I_{\mathcal H_{\rm q}}.
\label{eq:S3_isometry}
\end{equation}

For arbitrary logical states
\(|\psi_L\rangle,|\varphi_L\rangle\in\mathcal C\), their embedded
states satisfy

\begin{align}
\langle\widetilde\varphi_L|\widetilde\psi_L\rangle
&=
(\Phi|\varphi_L\rangle)^\dagger
(\Phi|\psi_L\rangle)
\nonumber\\
&=
\langle\varphi_L|
\Phi^\dagger\Phi
|\psi_L\rangle
\nonumber\\
&=
\langle\varphi_L|\psi_L\rangle.
\label{eq:S3_innerproduct}
\end{align}

Thus, every inner product between logical states is unchanged.
Consequently, normalization, orthogonality, transition amplitudes,
relative phases, and quantum coherences are identical in
\(\mathcal C\) and \(\widetilde{\mathcal C}\). The embedding therefore
changes only the physical representation of the code, not the quantum
information encoded in it.

We next prove that every embedded codeword has a unique logical
preimage. Existence follows directly from the definition of
\(\widetilde{\mathcal C}\): for every
\(|\widetilde\psi_L\rangle\in\widetilde{\mathcal C}\), there exists at
least one \(|\psi_L\rangle\in\mathcal C\) such that
\(|\widetilde\psi_L\rangle=\Phi|\psi_L\rangle\). To prove uniqueness,
suppose that the same embedded state arises from two logical states,

\begin{equation}
\Phi|\psi_L\rangle
=
\Phi|\varphi_L\rangle.
\label{eq:S3_equal_images}
\end{equation}

Applying \(\Phi^\dagger\) to both sides and using
Eq.~(\ref{eq:S3_isometry}) yields

\begin{equation}
|\psi_L\rangle
=
\Phi^\dagger\Phi|\psi_L\rangle
=
\Phi^\dagger\Phi|\varphi_L\rangle
=
|\varphi_L\rangle.
\label{eq:S3_unique_preimage}
\end{equation}

Therefore, one embedded codeword cannot represent two distinct logical
states. Every state in \(\mathcal C\) has exactly one image in
\(\widetilde{\mathcal C}\), and every state in
\(\widetilde{\mathcal C}\) has exactly one preimage in \(\mathcal C\).
The embedding consequently establishes a one-to-one correspondence
between the two code spaces.

This correspondence also proves that their dimensions are equal. Let
\(\{|i_L\rangle\}_{i=0}^{2^k-1}\) be an orthonormal basis of
\(\mathcal C\), and define

\begin{equation}
|\widetilde i_L\rangle
=
\Phi|i_L\rangle.
\label{eq:S3_embedded_basis_definition}
\end{equation}

Equation~(\ref{eq:S3_innerproduct}) gives

\begin{equation}
\langle\widetilde i_L|\widetilde j_L\rangle
=
\langle i_L|
\Phi^\dagger\Phi
|j_L\rangle
=
\delta_{ij},
\label{eq:S3_embedded_basis_orthogonality}
\end{equation}

so the \(2^k\) embedded basis states are linearly independent. They also
span \(\widetilde{\mathcal C}\): every
\(|\widetilde\psi_L\rangle\in\widetilde{\mathcal C}\) originates from
some \(|\psi_L\rangle=\sum_i c_i|i_L\rangle\in\mathcal C\), and hence

\begin{equation}
|\widetilde\psi_L\rangle
=
\Phi|\psi_L\rangle
=
\sum_{i=0}^{2^k-1}
c_i|\widetilde i_L\rangle.
\label{eq:S3_embedded_basis_span}
\end{equation}

Therefore,
\(\{|\widetilde i_L\rangle\}_{i=0}^{2^k-1}\) is an orthonormal basis of
\(\widetilde{\mathcal C}\), and

\begin{equation}
\dim\widetilde{\mathcal C}
=
\dim\mathcal C
=
2^k.
\label{eq:S3_dimension}
\end{equation}

The embedded system thus stores exactly the same \(k\) logical qubits
as the original code.

Finally, if

\begin{equation}
P_{\mathcal C}
=
\sum_{i=0}^{2^k-1}
|i_L\rangle\langle i_L|
\end{equation}

projects onto \(\mathcal C\), then

\begin{align}
\widetilde P_{\mathcal C}
&=
\Phi P_{\mathcal C}\Phi^\dagger
\nonumber\\
&=
\sum_{i=0}^{2^k-1}
|\widetilde i_L\rangle
\langle\widetilde i_L|
\label{eq:S3_embedded_projector}
\end{align}

is the orthogonal projector onto \(\widetilde{\mathcal C}\).

We conclude that the original and embedded code spaces contain exactly
the same logical information. Every logical state has a unique
multilevel representation, all logical overlaps and coherences are
preserved, and both spaces encode the same number of logical qubits.
This exact equivalence provides the basis for transferring the
stabilizer algebra, logical operators, syndrome measurements, and
recovery operations to the multilevel system.
\subsection{S3.2 Preservation of Operator Relations}
\label{S3.2}

For any operator \(A\) acting on \(\mathcal H_{\rm q}\), define its
embedded representation by

\begin{equation}
\widetilde A
=
\Phi A\Phi^\dagger .
\label{eq:S3_operator_transfer}
\end{equation}

The essential property of this transfer is that operator products are
preserved. Indeed, using
\(\Phi^\dagger\Phi=I_{\mathcal H_{\rm q}}\),

\begin{equation}
\widetilde A\,\widetilde B
=
\Phi A
(\Phi^\dagger\Phi)
B\Phi^\dagger
=
\Phi AB\Phi^\dagger .
\label{eq:S3_product_preservation}
\end{equation}

Consequently, the commutator and anticommutator satisfy

\begin{align}
[\widetilde A,\widetilde B]
&=
\Phi[A,B]\Phi^\dagger,
\nonumber\\
\{\widetilde A,\widetilde B\}
&=
\Phi\{A,B\}\Phi^\dagger.
\label{eq:S3_commutation_preservation}
\end{align}

Therefore, if \(A\) and \(B\) commute in the qubit representation,
\([A,B]=0\), then their embedded counterparts also commute,
\([\widetilde A,\widetilde B]=0\). Likewise, if they anticommute,
\(\{A,B\}=0\), then
\(\{\widetilde A,\widetilde B\}=0\).

Thus, all commutation and anticommutation relations required by the
stabilizer formalism are reproduced exactly within
\(\mathcal H_{\rm emb}\). In particular, commuting stabilizer
generators remain commuting, while the anticommutation relations between
logical Pauli operators and detectable errors are unchanged.

\subsection{S3.3 Physical Consequences of Theorem 1}
\label{S3.3}

Theorem 1 of the main text establishes that the embedded code satisfies
the same Knill--Laflamme conditions as the original stabilizer code.
Consequently, the complete error-correction structure is preserved under
the embedding.

In particular, for any embedded logical basis
\(\{|\widetilde i_L\rangle\}\),

\begin{equation}
\langle\widetilde i_L|
\widetilde E_a^\dagger
\widetilde E_b
|\widetilde j_L\rangle
=
\alpha_{ab}\delta_{ij},
\label{eq:S3_KL_basis}
\end{equation}

showing that the distinguishability relations among correctable error
sectors are identical before and after embedding.

Therefore:

\begin{enumerate}
\item Correctable error sets remain correctable.
\item Degenerate codes remain degenerate.
\item Nondegenerate codes remain nondegenerate.
\item The syndrome structure is unchanged.
\item Recovery operations may be transferred directly to the embedded
system.
\end{enumerate}

Hence, the embedding preserves not merely the code space but the entire
quantum error-correction protocol.

\subsection{S3.4 Preservation of Stabilizer Generators}
\label{S3.4}

Let the stabilizer group of the original code be generated by the
independent commuting operators

\begin{equation}
\mathcal S
=
\langle
S_1,\ldots,S_{n-k}
\rangle.
\end{equation}

The code space is their simultaneous \(+1\)-eigenspace:

\begin{equation}
S_r|\psi_L\rangle
=
|\psi_L\rangle,
\qquad
r=1,\ldots,n-k.
\label{eq:S3_original_stabilization}
\end{equation}

Define the embedded stabilizer generators by

\begin{equation}
\widetilde S_r
=
\Phi S_r\Phi^\dagger.
\label{eq:S3_embedded_stabilizer}
\end{equation}

\begin{proposition}[Preservation of Stabilizer Generators]
\label{prop:S3_stabilizers}

The transferred generators satisfy

\begin{align}
\widetilde S_r^\dagger
&=
\widetilde S_r,
\label{eq:S3_stabilizer_hermitian}\\
[\widetilde S_r,\widetilde S_t]
&=
0,
\label{eq:S3_stabilizer_commuting}\\
\widetilde S_r^{\,2}
&=
P_{\rm emb},
\label{eq:S3_stabilizer_square}
\end{align}

and stabilize every embedded logical state:

\begin{equation}
\widetilde S_r
|\widetilde\psi_L\rangle
=
|\widetilde\psi_L\rangle.
\label{eq:S3_embedded_stabilization}
\end{equation}

Furthermore, the embedded code space is the simultaneous
\(+1\)-eigenspace of the transferred generators within
\(\mathcal H_{\rm emb}\).

\end{proposition}

\begin{proof}

Hermiticity follows from

\begin{equation}
\widetilde S_r^\dagger
=
\Phi S_r^\dagger\Phi^\dagger
=
\Phi S_r\Phi^\dagger
=
\widetilde S_r.
\end{equation}

Since the original stabilizer generators commute,

\begin{align}
[\widetilde S_r,\widetilde S_t]
&=
\Phi[S_r,S_t]\Phi^\dagger
\nonumber\\
&=
0.
\end{align}

Moreover, because \(S_r^2=I_{\mathcal H_{\rm q}}\),

\begin{align}
\widetilde S_r^{\,2}
&=
\Phi S_r^2\Phi^\dagger
\nonumber\\
&=
\Phi I_{\mathcal H_{\rm q}}\Phi^\dagger
\nonumber\\
&=
P_{\rm emb}.
\end{align}

For any
\(|\widetilde\psi_L\rangle=\Phi|\psi_L\rangle\),

\begin{align}
\widetilde S_r
|\widetilde\psi_L\rangle
&=
\Phi S_r\Phi^\dagger
\Phi|\psi_L\rangle
\nonumber\\
&=
\Phi S_r|\psi_L\rangle
\nonumber\\
&=
\Phi|\psi_L\rangle
\nonumber\\
&=
|\widetilde\psi_L\rangle.
\end{align}

Conversely, let
\(|\widetilde\chi\rangle\in\mathcal H_{\rm emb}\)
satisfy

\begin{equation}
\widetilde S_r|\widetilde\chi\rangle
=
|\widetilde\chi\rangle
\end{equation}

for all \(r\). Since
\(|\widetilde\chi\rangle\in\mathcal H_{\rm emb}\), there exists a
unique
\(|\chi\rangle\in\mathcal H_{\rm q}\)
such that
\(|\widetilde\chi\rangle=\Phi|\chi\rangle\).
Applying \(\Phi^\dagger\) gives

\begin{equation}
S_r|\chi\rangle
=
|\chi\rangle
\end{equation}

for all \(r\), and therefore
\(|\chi\rangle\in\mathcal C\).
Hence
\(|\widetilde\chi\rangle\in\widetilde{\mathcal C}\).

\end{proof}

The relation
\(\widetilde S_r^2=P_{\rm emb}\)
shows that each transferred stabilizer has eigenvalues
\(\pm1\) on \(\mathcal H_{\rm emb}\) and acts as zero on
\(\mathcal H_{\rm leak}\). Therefore, stabilizer measurements within
the computational manifold reproduce exactly the measurements of the
original qubit code.

\subsection{S3.5 Preservation of Logical Operators}
\label{S3.5}

For an \([[n,k,d]]\) stabilizer code, let

\begin{equation}
\overline X_\ell,
\qquad
\overline Z_\ell,
\qquad
\ell=1,\ldots,k,
\end{equation}

denote a set of logical Pauli operators. Their embedded images are

\begin{equation}
\widetilde{\overline X}_{\ell}
=
\Phi\overline X_\ell\Phi^\dagger,
\qquad
\widetilde{\overline Z}_{\ell}
=
\Phi\overline Z_\ell\Phi^\dagger.
\label{eq:S3_embedded_logicals}
\end{equation}

\begin{proposition}[Preservation of Logical Operators]
\label{prop:S3_logicals}

The transferred logical operators preserve the embedded code space and
satisfy the same logical Pauli algebra as the original logical
operators. In particular,

\begin{align}
\widetilde{\overline X}_{\ell}^{\,2}
&=
P_{\rm emb},
&
\widetilde{\overline Z}_{\ell}^{\,2}
&=
P_{\rm emb},
\label{eq:S3_logical_square}\\
\widetilde{\overline X}_{\ell}
\widetilde{\overline Z}_{m}
&=
(-1)^{\delta_{\ell m}}
\widetilde{\overline Z}_{m}
\widetilde{\overline X}_{\ell},
\label{eq:S3_logical_Pauli}\\
[\widetilde{\overline X}_{\ell},\widetilde S_r]
&=
0,
&
[\widetilde{\overline Z}_{\ell},\widetilde S_r]
&=
0.
\label{eq:S3_logical_stabilizer_commutation}
\end{align}

Their action on embedded logical states is the embedded image of the
original logical action.

\end{proposition}

\begin{proof}

The logical Pauli relations follow directly from
Eq.~(\ref{eq:S3_product_preservation}). For example,

\begin{align}
\widetilde{\overline X}_{\ell}
\widetilde{\overline Z}_{m}
&=
\Phi
\overline X_\ell
\overline Z_m
\Phi^\dagger
\nonumber\\
&=
(-1)^{\delta_{\ell m}}
\Phi
\overline Z_m
\overline X_\ell
\Phi^\dagger
\nonumber\\
&=
(-1)^{\delta_{\ell m}}
\widetilde{\overline Z}_{m}
\widetilde{\overline X}_{\ell}.
\end{align}

Similarly,

\begin{align}
[\widetilde{\overline X}_{\ell},\widetilde S_r]
&=
\Phi
[\overline X_\ell,S_r]
\Phi^\dagger
\nonumber\\
&=
0,
\end{align}

and likewise for
\(\widetilde{\overline Z}_{\ell}\).

For an embedded logical state
\(|\widetilde\psi_L\rangle=\Phi|\psi_L\rangle\),

\begin{align}
\widetilde{\overline X}_{\ell}
|\widetilde\psi_L\rangle
&=
\Phi\overline X_\ell\Phi^\dagger
\Phi|\psi_L\rangle
\nonumber\\
&=
\Phi
\overline X_\ell|\psi_L\rangle.
\label{eq:S3_logical_action}
\end{align}

Thus, the transferred logical operator produces precisely the embedded
image of the corresponding original logical transformation.

\end{proof}

\subsection{S3.6 Preservation of Syndrome Relations and Measurements}
\label{S3.6}

The syndrome of an error is determined by its commutation or
anticommutation relations with the stabilizer generators. Suppose that

\begin{equation}
S_rE_a
=
(-1)^{s_r(a)}
E_aS_r,
\qquad
s_r(a)\in\{0,1\}.
\label{eq:S3_original_syndrome_relation}
\end{equation}

The syndrome vector is

\begin{equation}
\mathbf s(E_a)
=
\big(
s_1(a),\ldots,s_{n-k}(a)
\big).
\end{equation}

\begin{proposition}[Preservation of Syndrome Assignments]
\label{prop:S3_syndromes}

For the embedded error

\begin{equation}
\widetilde E_a
=
\Phi E_a\Phi^\dagger,
\end{equation}

one has

\begin{equation}
\widetilde S_r\widetilde E_a
=
(-1)^{s_r(a)}
\widetilde E_a\widetilde S_r.
\label{eq:S3_embedded_syndrome_relation}
\end{equation}

Consequently,

\begin{equation}
\mathbf s(\widetilde E_a)
=
\mathbf s(E_a).
\label{eq:S3_same_syndrome}
\end{equation}

\end{proposition}

\begin{proof}

Using operator-product preservation,

\begin{align}
\widetilde S_r\widetilde E_a
&=
\Phi S_rE_a\Phi^\dagger
\nonumber\\
&=
(-1)^{s_r(a)}
\Phi E_aS_r\Phi^\dagger
\nonumber\\
&=
(-1)^{s_r(a)}
\widetilde E_a\widetilde S_r.
\end{align}

Therefore, every syndrome bit is unchanged.

\end{proof}

For a single stabilizer generator \(S_r\), the original two-outcome
measurement projectors are

\begin{equation}
\Pi_r^{(\pm)}
=
\frac{I_{\mathcal H_{\rm q}}\pm S_r}{2}.
\label{eq:S3_original_single_projectors}
\end{equation}

Their embedded images are

\begin{align}
\widetilde\Pi_r^{(\pm)}
&=
\Phi\Pi_r^{(\pm)}\Phi^\dagger
\nonumber\\
&=
\frac{
P_{\rm emb}\pm\widetilde S_r
}{2}.
\label{eq:S3_embedded_single_projectors}
\end{align}

These operators satisfy

\begin{align}
\widetilde\Pi_r^{(+)}
+
\widetilde\Pi_r^{(-)}
&=
P_{\rm emb},
\label{eq:S3_single_complete}\\
\widetilde\Pi_r^{(+)}
\widetilde\Pi_r^{(-)}
&=
0.
\end{align}

Thus, they form a complete two-outcome projective measurement on the
embedded computational manifold.

For a normalized embedded state

\begin{equation}
\widetilde\rho
=
\Phi\rho\Phi^\dagger,
\end{equation}

the outcome probabilities are

\begin{align}
\widetilde p_r(\pm)
&=
\operatorname{Tr}
\left[
\widetilde\Pi_r^{(\pm)}
\widetilde\rho
\right]
\nonumber\\
&=
\operatorname{Tr}
\left[
\Phi\Pi_r^{(\pm)}
\rho
\Phi^\dagger
\right]
\nonumber\\
&=
\operatorname{Tr}
\left[
\Pi_r^{(\pm)}\rho
\right]
\nonumber\\
&=
p_r(\pm).
\label{eq:S3_measurement_probability}
\end{align}

Therefore, stabilizer-measurement statistics are preserved exactly.

If the physical system contains a nontrivial leakage sector, then a
measurement complete on the full Hilbert space
\(\mathcal H_D\)
may be written as the three-outcome measurement

\begin{equation}
\left\{
\widetilde\Pi_r^{(+)},
\widetilde\Pi_r^{(-)},
P_{\rm leak}
\right\},
\qquad
P_{\rm leak}
=
I_{\mathcal H_D}-P_{\rm emb}.
\label{eq:S3_full_measurement}
\end{equation}

The additional outcome \(P_{\rm leak}\) distinguishes population
outside the embedded computational manifold from ordinary stabilizer
syndromes.

\subsection{S3.7 Preservation of Syndrome Projectors}
\label{S3.7}

For the full syndrome vector

\begin{equation}
\mathbf s
=
(s_1,\ldots,s_{n-k}),
\end{equation}

the corresponding qubit-space syndrome projector is

\begin{equation}
P_{\mathbf s}
=
\prod_{r=1}^{n-k}
\frac{
I_{\mathcal H_{\rm q}}
+
(-1)^{s_r}S_r
}{2}.
\label{eq:S3_original_syndrome_projector}
\end{equation}

Its embedded counterpart is

\begin{equation}
\widetilde P_{\mathbf s}
=
\Phi
P_{\mathbf s}
\Phi^\dagger.
\label{eq:S3_embedded_syndrome_projector}
\end{equation}

Using the operator-transfer property repeatedly yields

\begin{equation}
\widetilde P_{\mathbf s}
=
\prod_{r=1}^{n-k}
\frac{
P_{\rm emb}
+
(-1)^{s_r}
\widetilde S_r
}{2}.
\label{eq:S3_embedded_syndrome_product}
\end{equation}

The embedded syndrome projectors satisfy

\begin{align}
\widetilde P_{\mathbf s}^{\,\dagger}
&=
\widetilde P_{\mathbf s},
\\
\widetilde P_{\mathbf s}^{\,2}
&=
\widetilde P_{\mathbf s},
\\
\widetilde P_{\mathbf s}
\widetilde P_{\mathbf t}
&=
\delta_{\mathbf s,\mathbf t}
\widetilde P_{\mathbf s},
\label{eq:S3_syndrome_orthogonality}\\
\sum_{\mathbf s}
\widetilde P_{\mathbf s}
&=
P_{\rm emb}.
\label{eq:S3_syndrome_completeness}
\end{align}

Hence, the syndrome decomposition of the original qubit Hilbert space
is transferred exactly to a decomposition of the embedded
computational manifold.

Moreover, if a correctable error \(E_a\) has syndrome
\(\mathbf s(a)\), then

\begin{equation}
P_{\mathbf s}
E_aP_{\mathcal C}
=
\delta_{\mathbf s,\mathbf s(a)}
E_aP_{\mathcal C}.
\label{eq:S3_original_error_sector}
\end{equation}

Embedding this relation gives

\begin{equation}
\widetilde P_{\mathbf s}
\widetilde E_a
\widetilde P_{\mathcal C}
=
\delta_{\mathbf s,\mathbf s(a)}
\widetilde E_a
\widetilde P_{\mathcal C}.
\label{eq:S3_embedded_error_sector}
\end{equation}

Thus, each correctable error is assigned to precisely the same syndrome
sector before and after embedding.

\subsection{S3.8 Preservation of Recovery Operations}
\label{S3.8}

We now establish that the complete recovery procedure is transferred
to the multilevel realization. We first consider syndrome-conditioned
recovery operators and then formulate the result for general quantum
channels.

Let \(R_{\mathbf s}\) be the corrective operation associated with the
syndrome \(\mathbf s\), and define

\begin{equation}
\widetilde R_{\mathbf s}
=
\Phi
R_{\mathbf s}
\Phi^\dagger.
\label{eq:S3_embedded_recovery_operator}
\end{equation}

Suppose that error \(E_a\) has syndrome
\(\mathbf s(a)\) and satisfies

\begin{equation}
R_{\mathbf s(a)}
E_a|\psi_L\rangle
=
|\psi_L\rangle
\label{eq:S3_original_operator_recovery}
\end{equation}

for every
\(|\psi_L\rangle\in\mathcal C\).
Then

\begin{align}
\widetilde R_{\mathbf s(a)}
\widetilde E_a
|\widetilde\psi_L\rangle
&=
\Phi
R_{\mathbf s(a)}
E_a
\Phi^\dagger\Phi
|\psi_L\rangle
\nonumber\\
&=
\Phi
R_{\mathbf s(a)}
E_a
|\psi_L\rangle
\nonumber\\
&=
\Phi|\psi_L\rangle
\nonumber\\
&=
|\widetilde\psi_L\rangle.
\label{eq:S3_embedded_operator_recovery}
\end{align}

Therefore, syndrome-conditioned correction is preserved exactly.

For a general recovery channel, let
\(\mathcal R\)
be a completely positive trace-preserving map on
\(\mathcal H_{\rm q}\)
that corrects the original noise channel
\(\mathcal E\) on the code:

\begin{equation}
(\mathcal R\circ\mathcal E)(\rho_L)
=
\rho_L
\label{eq:S3_original_channel_recovery}
\end{equation}

for every logical density operator
\(\rho_L=P_{\mathcal C}\rho_LP_{\mathcal C}\).

On operators supported in
\(\mathcal H_{\rm emb}\), define the transferred recovery map by

\begin{equation}
\widetilde{\mathcal R}_{\rm emb}(X)
=
\Phi
\mathcal R
\left(
\Phi^\dagger X\Phi
\right)
\Phi^\dagger.
\label{eq:S3_compressed_recovery}
\end{equation}

This map is completely positive and trace-preserving on states
supported in \(\mathcal H_{\rm emb}\). To obtain a channel defined on
the entire physical Hilbert space, choose an arbitrary normalized state
\(\tau\in\mathcal D(\mathcal H_D)\) and define

\begin{equation}
\widetilde{\mathcal R}(X)
=
\Phi
\mathcal R
\left(
\Phi^\dagger X\Phi
\right)
\Phi^\dagger
+
\operatorname{Tr}
\left(
P_{\rm leak}X
\right)
\tau.
\label{eq:S3_full_recovery_channel}
\end{equation}

The first term recovers the computational component, whereas the second
term specifies a valid action on the leakage sector. The map
\(\widetilde{\mathcal R}\) is completely positive. Furthermore,

\begin{align}
\operatorname{Tr}
\left[
\widetilde{\mathcal R}(X)
\right]
&=
\operatorname{Tr}
\left[
\Phi^\dagger X\Phi
\right]
+
\operatorname{Tr}
\left[
P_{\rm leak}X
\right]
\nonumber\\
&=
\operatorname{Tr}
\left[
P_{\rm emb}X
\right]
+
\operatorname{Tr}
\left[
P_{\rm leak}X
\right]
\nonumber\\
&=
\operatorname{Tr}(X),
\end{align}

so it is trace preserving on
\(\mathcal H_D\).

Define the embedded noise channel on the computational manifold by

\begin{equation}
\widetilde{\mathcal E}_{\rm emb}(X)
=
\Phi
\mathcal E
\left(
\Phi^\dagger X\Phi
\right)
\Phi^\dagger.
\label{eq:S3_embedded_noise_channel}
\end{equation}

For an embedded logical state

\begin{equation}
\widetilde\rho_L
=
\Phi\rho_L\Phi^\dagger,
\end{equation}

one obtains

\begin{align}
&
\widetilde{\mathcal R}
\left[
\widetilde{\mathcal E}_{\rm emb}
\left(
\widetilde\rho_L
\right)
\right]
\nonumber\\
&\quad=
\Phi
\mathcal R
\left[
\mathcal E(\rho_L)
\right]
\Phi^\dagger
\nonumber\\
&\quad=
\Phi\rho_L\Phi^\dagger
\nonumber\\
&\quad=
\widetilde\rho_L.
\label{eq:S3_embedded_channel_recovery}
\end{align}

The leakage-completion term in
Eq.~(\ref{eq:S3_full_recovery_channel}) does not contribute because
the embedded error state is supported entirely in
\(\mathcal H_{\rm emb}\).

Therefore, every exact recovery procedure for the original code
induces an exact recovery procedure for the embedded code. The action
outside the computational manifold may be chosen independently without
altering the correction of embedded errors.

\subsection{S3.9 Complete Equivalence of the Embedded Code}
\label{S3.9}

The preceding results establish the following complete equivalence:

\begin{equation}
\boxed{
\begin{aligned}
\mathcal C
&\longleftrightarrow
\widetilde{\mathcal C},
\\
P_{\mathcal C}
&\longleftrightarrow
\widetilde P_{\mathcal C},
\\
E_a
&\longleftrightarrow
\widetilde E_a,
\\
S_r
&\longleftrightarrow
\widetilde S_r,
\\
\overline X_\ell,\overline Z_\ell
&\longleftrightarrow
\widetilde{\overline X}_{\ell},
\widetilde{\overline Z}_{\ell},
\\
P_{\mathbf s}
&\longleftrightarrow
\widetilde P_{\mathbf s},
\\
R_{\mathbf s}
&\longleftrightarrow
\widetilde R_{\mathbf s}.
\end{aligned}
}
\label{eq:S3_complete_correspondence}
\end{equation}

The isometric property of \(\Phi\) preserves the geometry of the
logical code space, while the operator-transfer map preserves the full
\(*\)-algebra of the stabilizer code. Consequently, the embedding
preserves the Knill--Laflamme conditions, stabilizer generators,
logical Pauli relations, syndrome assignments, syndrome-measurement
statistics, syndrome projectors, and recovery operations.

The embedded multilevel realization is therefore mathematically and
operationally equivalent to the original stabilizer-code
implementation within the computational manifold. Any additional
levels of the physical system constitute a distinct leakage sector and
do not modify the transferred qubit-code structure.
\section{S4. Applications to Representative Stabilizer Codes}
\label{S4}

The examples presented in this section provide explicit realizations of
the embedding framework for three representative stabilizer codes of
increasing complexity. Together, they illustrate the preservation of
the logical code space, stabilizer generators, syndrome structure, and
recovery operations established in Sec.~S3. The first example
demonstrates the construction for the simplest nontrivial stabilizer
code, the second verifies the framework for a general nondegenerate
stabilizer code, and the third illustrates preservation of a CSS
stabilizer structure.
\subsection{S4.1 Example 1: Three-qubit Repetition Code}

This example provides a complete verification of the embedding
construction for the simplest nontrivial stabilizer code. Every stage of
the embedding procedure is carried out explicitly, beginning with the
binary-to-decimal basis mapping and proceeding through the construction
of the embedded logical states, stabilizer generators, error operators,
syndrome relations, and recovery operations. The purpose of this example
is to demonstrate, step by step, how the general embedding framework
developed in the main text reproduces the complete quantum
error-correction protocol within a single eight-level quantum system.

The three-qubit repetition code encodes one logical qubit into three
physical qubits with logical basis

\begin{equation}
|0_L\rangle=|000\rangle,
\qquad
|1_L\rangle=|111\rangle.
\label{eq:SM_three_codewords}
\end{equation}

The stabilizer generators are

\begin{equation}
S_1=Z_1Z_2,
\qquad
S_2=Z_2Z_3,
\label{eq:SM_three_stabilizers}
\end{equation}

and one possible choice of logical operators is

\begin{equation}
\overline X=X_1X_2X_3,
\qquad
\overline Z=Z_1.
\label{eq:SM_three_logicals}
\end{equation}

The code corrects the single-bit-flip error set

\begin{equation}
\mathcal E
=
\{
I,
X_1,
X_2,
X_3
\}.
\label{eq:SM_three_errors}
\end{equation}

\paragraph{Embedded logical codewords.}

For $D=8$, the embedding map

\begin{equation}
\Phi
=
\sum_{m=0}^{7}
|m\rangle\langle b(m)|
\end{equation}

acts as

\begin{align}
\Phi|000\rangle&=|0\rangle,\\
\Phi|111\rangle&=|7\rangle.
\end{align}

Therefore

\begin{equation}
|\widetilde0_L\rangle
=
\Phi|0_L\rangle
=
|0\rangle,
\qquad
|\widetilde1_L\rangle
=
\Phi|1_L\rangle
=
|7\rangle.
\label{eq:SM_three_embedded_codewords}
\end{equation}

An arbitrary logical state

\begin{equation}
|\psi_L\rangle
=
\alpha|0_L\rangle
+
\beta|1_L\rangle
\end{equation}

is transferred to

\begin{equation}
|\widetilde\psi_L\rangle
=
\alpha|0\rangle
+
\beta|7\rangle.
\end{equation}

\paragraph{Construction and measurement of the embedded stabilizer generators.}

We now show explicitly that the transferred operators
\(\widetilde S_1\) and \(\widetilde S_2\) are the stabilizer generators
of the embedded code. The original generators are

\begin{equation}
S_1=Z_1Z_2=Z\otimes Z\otimes I,
\qquad
S_2=Z_2Z_3=I\otimes Z\otimes Z.
\label{eq:SM_three_original_stabilizers}
\end{equation}

They are transferred to the eight-level realization according to

\begin{equation}
\widetilde S_r
=
\Phi S_r\Phi^\dagger,
\qquad
r=1,2.
\label{eq:SM_three_embedded_stabilizers}
\end{equation}

We use the ordered computational basis

\begin{equation}
\mathcal B_q
=
\{
|000\rangle,
|001\rangle,
|010\rangle,
|011\rangle,
|100\rangle,
|101\rangle,
|110\rangle,
|111\rangle
\},
\end{equation}

and the ordered basis of the eight-level system

\begin{equation}
\mathcal B_8
=
\{
|0\rangle,
|1\rangle,
|2\rangle,
|3\rangle,
|4\rangle,
|5\rangle,
|6\rangle,
|7\rangle
\}.
\end{equation}

The embedding map establishes the correspondence

\begin{align}
|000\rangle&\mapsto|0\rangle,
&
|001\rangle&\mapsto|1\rangle,
&
|010\rangle&\mapsto|2\rangle,
&
|011\rangle&\mapsto|3\rangle,
\nonumber\\
|100\rangle&\mapsto|4\rangle,
&
|101\rangle&\mapsto|5\rangle,
&
|110\rangle&\mapsto|6\rangle,
&
|111\rangle&\mapsto|7\rangle.
\label{eq:SM_three_basis_correspondence}
\end{align}

With respect to these ordered bases, the matrix of \(\Phi\) is \(I_8\).
This does not mean that the physical systems are identical: the columns
of \(\Phi\) refer to basis states of three qubits, whereas its rows refer
to the levels of one eight-dimensional system. The map therefore
identifies the two basis representations while changing the physical
realization.

\subparagraph{First embedded stabilizer.}

In the computational basis \(\mathcal B_q\),

\begin{equation}
S_1
=
Z\otimes Z\otimes I
=
\begin{pmatrix}
 1&0&0&0&0&0&0&0\\
 0&1&0&0&0&0&0&0\\
 0&0&-1&0&0&0&0&0\\
 0&0&0&-1&0&0&0&0\\
 0&0&0&0&-1&0&0&0\\
 0&0&0&0&0&-1&0&0\\
 0&0&0&0&0&0&1&0\\
 0&0&0&0&0&0&0&1
\end{pmatrix}.
\label{eq:SM_three_S1_matrix}
\end{equation}

Consequently,

\begin{align}
\widetilde S_1
&=
\Phi S_1\Phi^\dagger
\nonumber\\
&=
I_8
\begin{pmatrix}
 1&0&0&0&0&0&0&0\\
 0&1&0&0&0&0&0&0\\
 0&0&-1&0&0&0&0&0\\
 0&0&0&-1&0&0&0&0\\
 0&0&0&0&-1&0&0&0\\
 0&0&0&0&0&-1&0&0\\
 0&0&0&0&0&0&1&0\\
 0&0&0&0&0&0&0&1
\end{pmatrix}
I_8
\nonumber\\
&=
\operatorname{diag}
(1,1,-1,-1,-1,-1,1,1).
\label{eq:SM_three_embedded_S1}
\end{align}

As an operator on the eight-level system, this may be written as

\begin{align}
\widetilde S_1
={}&
|0\rangle\langle0|
+
|1\rangle\langle1|
-
|2\rangle\langle2|
-
|3\rangle\langle3|
\nonumber\\
&-
|4\rangle\langle4|
-
|5\rangle\langle5|
+
|6\rangle\langle6|
+
|7\rangle\langle7|.
\label{eq:SM_three_S1_spectral}
\end{align}

Thus, the \(+1\) and \(-1\) eigenspaces of
\(\widetilde S_1\) are

\begin{align}
\mathcal H_{1,+}
&=
\operatorname{span}
\{
|0\rangle,|1\rangle,|6\rangle,|7\rangle
\},
\nonumber\\
\mathcal H_{1,-}
&=
\operatorname{span}
\{
|2\rangle,|3\rangle,|4\rangle,|5\rangle
\}.
\end{align}

\subparagraph{Second embedded stabilizer.}

Similarly, in the computational basis,

\begin{equation}
S_2
=
I\otimes Z\otimes Z
=
\begin{pmatrix}
 1&0&0&0&0&0&0&0\\
 0&-1&0&0&0&0&0&0\\
 0&0&-1&0&0&0&0&0\\
 0&0&0&1&0&0&0&0\\
 0&0&0&0&1&0&0&0\\
 0&0&0&0&0&-1&0&0\\
 0&0&0&0&0&0&-1&0\\
 0&0&0&0&0&0&0&1
\end{pmatrix}.
\label{eq:SM_three_S2_matrix}
\end{equation}

Therefore,

\begin{align}
\widetilde S_2
&=
\Phi S_2\Phi^\dagger
\nonumber\\
&=
\operatorname{diag}
(1,-1,-1,1,1,-1,-1,1).
\label{eq:SM_three_embedded_S2}
\end{align}

Equivalently,

\begin{align}
\widetilde S_2
={}&
|0\rangle\langle0|
-
|1\rangle\langle1|
-
|2\rangle\langle2|
+
|3\rangle\langle3|
\nonumber\\
&+
|4\rangle\langle4|
-
|5\rangle\langle5|
-
|6\rangle\langle6|
+
|7\rangle\langle7|.
\label{eq:SM_three_S2_spectral}
\end{align}

Its eigenspaces are

\begin{align}
\mathcal H_{2,+}
&=
\operatorname{span}
\{
|0\rangle,|3\rangle,|4\rangle,|7\rangle
\},
\nonumber\\
\mathcal H_{2,-}
&=
\operatorname{span}
\{
|1\rangle,|2\rangle,|5\rangle,|6\rangle
\}.
\end{align}

\subparagraph{Verification of the stabilizer properties.}

The transferred operators satisfy all defining properties of stabilizer
generators. First, they are Hermitian,

\begin{equation}
\widetilde S_r^\dagger
=
\widetilde S_r,
\qquad r=1,2,
\end{equation}

and square to the identity,

\begin{equation}
\widetilde S_r^2
=
I_8.
\end{equation}

Their eigenvalues are therefore restricted to \(\pm1\). Moreover, since
the original stabilizers commute,

\begin{align}
[\widetilde S_1,\widetilde S_2]
&=
\Phi[S_1,S_2]\Phi^\dagger
\nonumber\\
&=
0.
\end{align}

Hence \(\widetilde S_1\) and \(\widetilde S_2\) are compatible
observables and may be used jointly for syndrome extraction.

Most importantly, the common \(+1\) eigenspace is

\begin{align}
\mathcal H_{1,+}\cap\mathcal H_{2,+}
&=
\operatorname{span}
\{
|0\rangle,|7\rangle
\}
\nonumber\\
&=
\operatorname{span}
\{
|\widetilde0_L\rangle,
|\widetilde1_L\rangle
\}
\nonumber\\
&\equiv
\widetilde{\mathcal C}.
\label{eq:SM_three_common_positive_eigenspace}
\end{align}

Thus, the embedded code space is exactly the simultaneous \(+1\)
eigenspace of the two transferred generators. This is the defining
property of a stabilizer code and proves that
\(\widetilde S_1\) and \(\widetilde S_2\) are the stabilizer generators
of \(\widetilde{\mathcal C}\).

The same result follows directly by acting on the embedded logical
basis:

\begin{align}
\widetilde S_1|\widetilde0_L\rangle
&=
\widetilde S_1|0\rangle
=
|0\rangle,
&
\widetilde S_2|\widetilde0_L\rangle
&=
\widetilde S_2|0\rangle
=
|0\rangle,
\nonumber\\
\widetilde S_1|\widetilde1_L\rangle
&=
\widetilde S_1|7\rangle
=
|7\rangle,
&
\widetilde S_2|\widetilde1_L\rangle
&=
\widetilde S_2|7\rangle
=
|7\rangle.
\end{align}

Consequently, for an arbitrary embedded logical state

\begin{equation}
|\widetilde\psi_L\rangle
=
\alpha|0\rangle+\beta|7\rangle,
\end{equation}

we obtain

\begin{equation}
\widetilde S_r|\widetilde\psi_L\rangle
=
|\widetilde\psi_L\rangle,
\qquad
r=1,2.
\label{eq:SM_three_embedded_state_stabilized}
\end{equation}

\subparagraph{Projector onto the embedded code space.}

The projectors associated with the outcomes \(\pm1\) of the embedded
stabilizer measurements are

\begin{equation}
\widetilde\Pi_r^{(\pm)}
=
\frac{I_8\pm\widetilde S_r}{2}.
\label{eq:SM_three_stabilizer_measurement_projectors}
\end{equation}

The joint \(+1\) projector is therefore

\begin{align}
\widetilde P_{\mathcal C}
&=
\widetilde\Pi_1^{(+)}
\widetilde\Pi_2^{(+)}
\nonumber\\
&=
\frac{1}{4}
\left(I_8+\widetilde S_1\right)
\left(I_8+\widetilde S_2\right)
\nonumber\\
&=
\operatorname{diag}(1,0,0,0,0,0,0,1)
\nonumber\\
&=
|0\rangle\langle0|
+
|7\rangle\langle7|.
\label{eq:SM_three_embedded_code_projector}
\end{align}

This is precisely the projector onto the embedded logical code space,

\begin{equation}
\widetilde P_{\mathcal C}
=
|\widetilde0_L\rangle
\langle\widetilde0_L|
+
|\widetilde1_L\rangle
\langle\widetilde1_L|.
\end{equation}

Equation~\eqref{eq:SM_three_embedded_code_projector} gives a direct
matrix-level verification that the transferred operators define the
correct embedded stabilizer code.

\subparagraph{Measurement and syndrome interpretation.}

In the original three-qubit realization, syndrome extraction corresponds
to measuring the Pauli products \(Z_1Z_2\) and \(Z_2Z_3\). In the
eight-level realization, the corresponding measurements are the
two-outcome projective measurements

\begin{equation}
\left\{
\widetilde\Pi_r^{(+)},
\widetilde\Pi_r^{(-)}
\right\},
\qquad r=1,2.
\end{equation}

For a density operator \(\widetilde\rho\), the probabilities of the two
outcomes are

\begin{equation}
p_r(\pm)
=
\operatorname{Tr}
\left[
\widetilde\Pi_r^{(\pm)}
\widetilde\rho
\right].
\end{equation}

For every state in the embedded code space,

\begin{equation}
\widetilde P_{\mathcal C}
\widetilde\rho_L
\widetilde P_{\mathcal C}
=
\widetilde\rho_L,
\end{equation}

and therefore

\begin{equation}
p_1(+)=p_2(+)=1.
\end{equation}

As an explicit error example, consider

\begin{equation}
|\widetilde\psi_L\rangle
=
\alpha|0\rangle+\beta|7\rangle.
\end{equation}

A transferred \(X_1\) error gives

\begin{equation}
\widetilde X_1|\widetilde\psi_L\rangle
=
\alpha|4\rangle+\beta|3\rangle.
\label{eq:SM_three_X1_error_state}
\end{equation}

Using the matrices above,

\begin{align}
\widetilde S_1
\left(
\alpha|4\rangle+\beta|3\rangle
\right)
&=
-
\left(
\alpha|4\rangle+\beta|3\rangle
\right),
\nonumber\\
\widetilde S_2
\left(
\alpha|4\rangle+\beta|3\rangle
\right)
&=
+
\left(
\alpha|4\rangle+\beta|3\rangle
\right).
\end{align}

The measured syndrome is therefore

\begin{equation}
(-1,+1),
\end{equation}

which is exactly the syndrome of \(X_1\) in the original three-qubit
repetition code.

Similarly,

\begin{align}
\widetilde X_2|\widetilde\psi_L\rangle
&=
\alpha|2\rangle+\beta|5\rangle,
&
(\widetilde S_1,\widetilde S_2)
&=
(-1,-1),
\nonumber\\
\widetilde X_3|\widetilde\psi_L\rangle
&=
\alpha|1\rangle+\beta|6\rangle,
&
(\widetilde S_1,\widetilde S_2)
&=
(+1,-1).
\end{align}

Thus, the simultaneous eigenspaces of
\(\widetilde S_1\) and \(\widetilde S_2\) reproduce the four syndrome
subspaces of the original repetition code. The embedded operators are
therefore not merely matrix representations of the original
stabilizers; they stabilize the embedded logical space, distinguish the
transferred error spaces, and provide the same syndrome information
required for recovery.
\paragraph{Embedded error operators and syndrome preservation.}

The correctable bit-flip errors are transferred to the eight-level
system according to

\begin{equation}
\widetilde X_j
=
\Phi X_j\Phi^\dagger,
\qquad
j=1,2,3.
\label{eq:SM_three_embedded_errors}
\end{equation}

In the computational basis, \(X_j\) flips the \(j\)th binary digit.
Under the binary-to-level correspondence, this becomes a permutation of
the eight multilevel basis states. For example, \(X_1\) flips the most
significant bit,

\begin{align}
|000\rangle&\longleftrightarrow|100\rangle,
&
|001\rangle&\longleftrightarrow|101\rangle,
\nonumber\\
|010\rangle&\longleftrightarrow|110\rangle,
&
|011\rangle&\longleftrightarrow|111\rangle.
\end{align}

Using the embedding

\[
|000\rangle\mapsto|0\rangle,\quad
|001\rangle\mapsto|1\rangle,\quad\ldots,\quad
|111\rangle\mapsto|7\rangle,
\]

these transformations become

\begin{equation}
|0\rangle\longleftrightarrow|4\rangle,\qquad
|1\rangle\longleftrightarrow|5\rangle,\qquad
|2\rangle\longleftrightarrow|6\rangle,\qquad
|3\rangle\longleftrightarrow|7\rangle.
\end{equation}

Therefore,

\begin{align}
\widetilde X_1
={}&
|4\rangle\langle0|
+
|5\rangle\langle1|
+
|6\rangle\langle2|
+
|7\rangle\langle3|
\nonumber\\
&+
|0\rangle\langle4|
+
|1\rangle\langle5|
+
|2\rangle\langle6|
+
|3\rangle\langle7|.
\label{eq:SM_three_X1_explicit}
\end{align}

Equivalently,

\begin{equation}
\widetilde X_1
=
\sum_{m=0}^{7}
|m\oplus4\rangle\langle m|,
\end{equation}

where \(\oplus\) denotes bitwise XOR. The remaining transferred
bit-flip operators are

\begin{align}
\widetilde X_2
&=
\sum_{m=0}^{7}
|m\oplus2\rangle\langle m|,
\nonumber\\
\widetilde X_3
&=
\sum_{m=0}^{7}
|m\oplus1\rangle\langle m|.
\label{eq:SM_three_X2_X3}
\end{align}

Their explicit actions are

\begin{align}
\widetilde X_2:\quad
&|0\rangle\leftrightarrow|2\rangle,\quad
|1\rangle\leftrightarrow|3\rangle,\quad
|4\rangle\leftrightarrow|6\rangle,\quad
|5\rangle\leftrightarrow|7\rangle,
\nonumber\\
\widetilde X_3:\quad
&|0\rangle\leftrightarrow|1\rangle,\quad
|2\rangle\leftrightarrow|3\rangle,\quad
|4\rangle\leftrightarrow|5\rangle,\quad
|6\rangle\leftrightarrow|7\rangle.
\end{align}

For an arbitrary embedded logical state

\begin{equation}
|\widetilde\psi_L\rangle
=
\alpha|0\rangle+\beta|7\rangle,
\end{equation}

the three errors generate the states

\begin{align}
\widetilde X_1|\widetilde\psi_L\rangle
&=
\alpha|4\rangle+\beta|3\rangle,
\nonumber\\
\widetilde X_2|\widetilde\psi_L\rangle
&=
\alpha|2\rangle+\beta|5\rangle,
\nonumber\\
\widetilde X_3|\widetilde\psi_L\rangle
&=
\alpha|1\rangle+\beta|6\rangle.
\label{eq:SM_three_embedded_error_states}
\end{align}

These are precisely the embedded images of
\(X_j|\psi_L\rangle\).

The syndrome relations are preserved algebraically. If

\begin{equation}
S_rX_j
=
(-1)^{s_r(j)}
X_jS_r,
\end{equation}

then, using \(\Phi^\dagger\Phi=I\),

\begin{align}
\widetilde S_r\widetilde X_j
&=
\Phi S_r\Phi^\dagger
\Phi X_j\Phi^\dagger
\nonumber\\
&=
\Phi S_rX_j\Phi^\dagger
\nonumber\\
&=
(-1)^{s_r(j)}
\Phi X_jS_r\Phi^\dagger
\nonumber\\
&=
(-1)^{s_r(j)}
\widetilde X_j\widetilde S_r.
\label{eq:SM_three_syndrome_preservation}
\end{align}

Thus, each transferred error has the same commutation or
anticommutation relation with the embedded stabilizers as the
corresponding qubit error has with the original stabilizers. The
syndrome assignments are therefore unchanged.

\begin{table}[h]
\caption{Syndromes of the three-qubit repetition code and its embedded
eight-level realization.}
\label{tab:SM_three_syndrome}
\begin{ruledtabular}
\begin{tabular}{cccc}
Error & Stabilizer 1 & Stabilizer 2 & Syndrome\\
\hline
\(I\) or \(\widetilde I\)
    & \(+1\) & \(+1\) & \(00\)\\
\(X_1\) or \(\widetilde X_1\)
    & \(-1\) & \(+1\) & \(10\)\\
\(X_2\) or \(\widetilde X_2\)
    & \(-1\) & \(-1\) & \(11\)\\
\(X_3\) or \(\widetilde X_3\)
    & \(+1\) & \(-1\) & \(01\)
\end{tabular}
\end{ruledtabular}
\end{table}

\paragraph{Embedded recovery operations.}

After the syndrome identifies the error sector, the corresponding
correction is transferred according to

\begin{equation}
\widetilde R_j
=
\Phi R_j\Phi^\dagger.
\label{eq:SM_three_transferred_recovery}
\end{equation}

For the three-qubit repetition code, \(R_j=X_j\). Hence,

\begin{equation}
\widetilde R_j
=
\widetilde X_j,
\qquad
j=1,2,3,
\end{equation}

while the trivial-syndrome sector requires no correction,

\begin{equation}
\widetilde R_0=I_8.
\end{equation}

Because each transferred bit-flip operator is Hermitian and involutory,

\begin{equation}
\widetilde X_j^\dagger
=
\widetilde X_j,
\qquad
\widetilde X_j^2
=
I_8,
\end{equation}

the correction exactly reverses the corresponding error:

\begin{align}
\widetilde R_j
\widetilde X_j
|\widetilde\psi_L\rangle
&=
\widetilde X_j^2
|\widetilde\psi_L\rangle
\nonumber\\
&=
|\widetilde\psi_L\rangle.
\label{eq:SM_three_recovery_identity}
\end{align}

For example,

\begin{align}
\widetilde X_1|\widetilde\psi_L\rangle
&=
\alpha|4\rangle+\beta|3\rangle,
\nonumber\\
\widetilde R_1
\left(
\alpha|4\rangle+\beta|3\rangle
\right)
&=
\alpha|0\rangle+\beta|7\rangle
\nonumber\\
&=
|\widetilde\psi_L\rangle.
\end{align}

Therefore, the embedding preserves not only the stabilizer eigenspaces
and syndrome labels, but also the complete error-correction operation:
each correctable qubit error is mapped to a distinct multilevel error
sector, identified by the same syndrome and reversed by the transferred
recovery operator.

\subsection{S4.2 Five-qubit perfect code}

The $[[5,1,3]]$ perfect code encodes one logical qubit into five
physical qubits and corrects every arbitrary error acting on any one
physical qubit \cite{Bennett1996,Laflamme1996}. It therefore provides a
more demanding illustration of the embedding construction than the
three-qubit repetition code. In particular, the example demonstrates
the transfer of non-diagonal stabilizer generators, arbitrary
single-qubit Pauli errors, sixteen distinct syndrome sectors, and the
corresponding recovery operations to a single thirty-two-level system.

A cyclic set of stabilizer generators is

\begin{align}
S_1&=XZZXI, &
S_2&=IXZZX,
\nonumber\\
S_3&=XIXZZ, &
S_4&=ZXIXZ.
\label{eq:SM_five_stabilizers}
\end{align}

A convenient choice of logical Pauli operators is

\begin{equation}
\overline X=XXXXX,
\qquad
\overline Z=ZZZZZ.
\label{eq:SM_five_logicals}
\end{equation}

A normalized logical-zero state is

\begin{align}
\ket{0_L}
=
\frac{1}{4}\big(&
 \ket{00000}
+\ket{10010}
+\ket{01001}
+\ket{10100}
+\ket{01010}
\nonumber\\
&-\ket{11011}
-\ket{00110}
-\ket{11000}
-\ket{11101}
-\ket{00011}
\nonumber\\
&-\ket{11110}
-\ket{01111}
-\ket{10001}
-\ket{01100}
-\ket{10111}
+\ket{00101}
\big),
\label{eq:SM_five_zero}
\end{align}

and the logical-one state is

\begin{equation}
\ket{1_L}
=
\overline X\ket{0_L}.
\label{eq:SM_five_one}
\end{equation}

The correctable Pauli family is

\begin{equation}
\mathcal E_5
=
\{I\}
\cup
\{X_j,Y_j,Z_j:j=1,\ldots,5\}.
\label{eq:SM_five_error_set}
\end{equation}

\paragraph{Embedding into a thirty-two-level system.}

For $D=32$, the embedding map is

\begin{equation}
\Phi
=
\sum_{m=0}^{31}
\ket{m}\bra{b(m)},
\label{eq:SM_five_phi}
\end{equation}

where $\ket{b(m)}$ denotes the five-qubit computational basis state
whose binary representation is the integer $m$. Explicitly,

\begin{equation}
\ket{b_1b_2b_3b_4b_5}
\longmapsto
\ket{
16b_1+8b_2+4b_3+2b_4+b_5
}.
\label{eq:SM_five_binary_decimal}
\end{equation}

Thus, for example,

\begin{align}
\ket{00000}&\longmapsto\ket{0},
&
\ket{10010}&\longmapsto\ket{18},
\nonumber\\
\ket{01001}&\longmapsto\ket{9},
&
\ket{10100}&\longmapsto\ket{20}.
\end{align}

Using the ordered computational basis

\begin{equation}
\mathcal B_q
=
\{
\ket{00000},\ket{00001},\ldots,\ket{11111}
\},
\end{equation}

and the ordered multilevel basis

\begin{equation}
\mathcal B_{32}
=
\{
\ket0,\ket1,\ldots,\ket{31}
\},
\end{equation}

the matrix representation of $\Phi$ is numerically $I_{32}$. Its
domain and codomain are nevertheless physically different: its columns
label states of five physical qubits, whereas its rows label the levels
of one thirty-two-dimensional system.

The isometric identities are

\begin{equation}
\Phi^\dagger\Phi
=
I_{\mathcal H_q},
\qquad
\Phi\Phi^\dagger
=
I_{32},
\label{eq:SM_five_isometry}
\end{equation}

where the second identity follows because the thirty-two-dimensional
physical space is exactly the embedded image of the five-qubit Hilbert
space in this example.

\paragraph{Embedded logical codewords.}

Applying $\Phi$ term by term to Eq.~\eqref{eq:SM_five_zero} gives

\begin{align}
\ket{\widetilde 0_L}
&=
\Phi\ket{0_L}
\nonumber\\
&=
\frac{1}{4}\big(
 \ket0
+\ket{18}
+\ket9
+\ket{20}
+\ket{10}
-\ket{27}
-\ket6
-\ket{24}
\nonumber\\
&\hspace{1.15cm}
-\ket{29}
-\ket3
-\ket{30}
-\ket{15}
-\ket{17}
-\ket{12}
-\ket{23}
+\ket5
\big).
\label{eq:SM_five_embedded_zero}
\end{align}

The transferred logical-$X$ operator is

\begin{equation}
\widetilde{\overline X}
=
\Phi\overline X\Phi^\dagger.
\end{equation}

Because $\overline X=X^{\otimes5}$ complements all five binary digits,

\begin{equation}
\widetilde{\overline X}
=
\sum_{m=0}^{31}
\ket{m\oplus31}\bra m,
\label{eq:SM_five_embedded_logical_X}
\end{equation}

where $31=(11111)_2$ and $\oplus$ denotes bitwise XOR. Therefore,

\begin{equation}
\ket{\widetilde1_L}
=
\Phi\ket{1_L}
=
\widetilde{\overline X}
\ket{\widetilde0_L}.
\label{eq:SM_five_embedded_one}
\end{equation}

Similarly,

\begin{equation}
\widetilde{\overline Z}
=
\Phi\overline Z\Phi^\dagger
=
\sum_{m=0}^{31}
(-1)^{b_1(m)+\cdots+b_5(m)}
\ket m\bra m.
\label{eq:SM_five_embedded_logical_Z}
\end{equation}

Consequently,

\begin{align}
\widetilde{\overline Z}\ket{\widetilde0_L}
&=
\ket{\widetilde0_L},
\nonumber\\
\widetilde{\overline Z}\ket{\widetilde1_L}
&=
-\ket{\widetilde1_L},
\end{align}

and

\begin{align}
\widetilde{\overline X}\ket{\widetilde0_L}
&=
\ket{\widetilde1_L},
\nonumber\\
\widetilde{\overline X}\ket{\widetilde1_L}
&=
\ket{\widetilde0_L}.
\end{align}

An arbitrary encoded state

\begin{equation}
\ket{\psi_L}
=
\alpha\ket{0_L}
+
\beta\ket{1_L}
\end{equation}

is therefore represented as

\begin{equation}
\ket{\widetilde\psi_L}
=
\Phi\ket{\psi_L}
=
\alpha\ket{\widetilde0_L}
+
\beta\ket{\widetilde1_L},
\label{eq:SM_five_general_embedded_state}
\end{equation}

with the logical amplitudes $\alpha$ and $\beta$ unchanged.

\paragraph{Transferred single-qubit Pauli operators.}

Before constructing the embedded stabilizers, it is useful to obtain
the exact multilevel representation of the single-qubit Pauli
operators. Let

\begin{equation}
w_j=2^{5-j},
\qquad
(w_1,w_2,w_3,w_4,w_5)
=
(16,8,4,2,1),
\label{eq:SM_five_binary_weights}
\end{equation}

and let $b_j(m)\in\{0,1\}$ denote the $j$th binary digit of $m$, with
$b_1$ the most significant digit.

A bit flip on virtual qubit $j$ complements only the corresponding
binary digit. Hence,

\begin{equation}
\widetilde X_j
=
\Phi X_j\Phi^\dagger
=
\sum_{m=0}^{31}
\ket{m\oplus w_j}\bra m.
\label{eq:SM_five_embedded_X}
\end{equation}

Its matrix elements in the multilevel basis are therefore

\begin{equation}
\bra n\widetilde X_j\ket m
=
\delta_{n,m\oplus w_j}.
\label{eq:SM_five_X_matrix_elements}
\end{equation}

Thus, $\widetilde X_j$ is an exact permutation matrix. For example,

\begin{align}
\widetilde X_1
={}&
\ket{16}\bra0+\ket{17}\bra1+\cdots+\ket{31}\bra{15}
\nonumber\\
&+
\ket0\bra{16}+\ket1\bra{17}+\cdots+\ket{15}\bra{31}.
\label{eq:SM_five_X1_explicit}
\end{align}

Accordingly,

\begin{equation}
\ket0
\stackrel{\widetilde X_1}{\longleftrightarrow}
\ket{16},
\qquad
\ket5
\stackrel{\widetilde X_1}{\longleftrightarrow}
\ket{21},
\end{equation}

which are precisely the embedded images of

\begin{equation}
\ket{00000}
\stackrel{X_1}{\longleftrightarrow}
\ket{10000},
\qquad
\ket{00101}
\stackrel{X_1}{\longleftrightarrow}
\ket{10101}.
\end{equation}

A phase flip on virtual qubit $j$ is diagonal in the computational
basis and produces the phase $(-1)^{b_j(m)}$. Its embedded
representation is therefore

\begin{equation}
\widetilde Z_j
=
\Phi Z_j\Phi^\dagger
=
\sum_{m=0}^{31}
(-1)^{b_j(m)}
\ket m\bra m,
\label{eq:SM_five_embedded_Z}
\end{equation}

with matrix elements

\begin{equation}
\bra n\widetilde Z_j\ket m
=
(-1)^{b_j(m)}
\delta_{nm}.
\label{eq:SM_five_Z_matrix_elements}
\end{equation}

Finally,

\begin{equation}
\widetilde Y_j
=
\Phi Y_j\Phi^\dagger
=
i\widetilde X_j\widetilde Z_j,
\label{eq:SM_five_embedded_Y}
\end{equation}

so that

\begin{equation}
\widetilde Y_j
=
i
\sum_{m=0}^{31}
(-1)^{b_j(m)}
\ket{m\oplus w_j}\bra m.
\label{eq:SM_five_Y_explicit}
\end{equation}

Thus, in the multilevel representation,

\[
\widetilde X_j
\text{ is a level permutation},\qquad
\widetilde Z_j
\text{ is diagonal},\qquad
\widetilde Y_j
\text{ is a signed level permutation}.
\]

Because the transfer is an isometric conjugation, the Pauli algebra is
preserved:

\begin{align}
\widetilde X_j^2
&=
\widetilde Y_j^2
=
\widetilde Z_j^2
=
I_{32},
\nonumber\\
\widetilde X_j\widetilde Z_j
&=
-\widetilde Z_j\widetilde X_j,
\nonumber\\
[\widetilde P_j,\widetilde Q_k]
&=
0,
\qquad
j\neq k,
\label{eq:SM_five_transferred_Pauli_algebra}
\end{align}

where $P,Q\in\{X,Y,Z\}$.

\paragraph{Construction of the embedded stabilizer generators.}

The stabilizer generators are transferred according to

\begin{equation}
\widetilde S_r
=
\Phi S_r\Phi^\dagger,
\qquad
r=1,\ldots,4.
\label{eq:SM_five_embedded_stabilizers_definition}
\end{equation}

Using Eq.~\eqref{eq:SM_five_isometry} between successive factors gives

\begin{align}
\widetilde S_1
&=
\widetilde X_1
\widetilde Z_2
\widetilde Z_3
\widetilde X_4,
\nonumber\\
\widetilde S_2
&=
\widetilde X_2
\widetilde Z_3
\widetilde Z_4
\widetilde X_5,
\nonumber\\
\widetilde S_3
&=
\widetilde X_1
\widetilde X_3
\widetilde Z_4
\widetilde Z_5,
\nonumber\\
\widetilde S_4
&=
\widetilde Z_1
\widetilde X_2
\widetilde X_4
\widetilde Z_5.
\label{eq:SM_five_embedded_stabilizers}
\end{align}

Unlike the stabilizers of the three-qubit repetition code, these
operators are not diagonal in the multilevel basis because each
generator contains two $X$ factors. They are instead signed permutation
matrices.

We now derive the first generator explicitly. Acting on a computational
basis state gives

\begin{equation}
S_1
\ket{b_1b_2b_3b_4b_5}
=
(-1)^{b_2+b_3}
\ket{
b_1\oplus1,\,
b_2,\,
b_3,\,
b_4\oplus1,\,
b_5
}.
\label{eq:SM_five_S1_computational_action}
\end{equation}

The first and fourth binary digits have weights $16$ and $2$,
respectively. Therefore, after embedding,

\begin{equation}
\widetilde S_1\ket m
=
(-1)^{b_2(m)+b_3(m)}
\ket{m\oplus16\oplus2}.
\label{eq:SM_five_S1_level_action}
\end{equation}

Since $16\oplus2=18$, the exact multilevel operator is

\begin{equation}
\boxed{
\widetilde S_1
=
\sum_{m=0}^{31}
(-1)^{b_2(m)+b_3(m)}
\ket{m\oplus18}\bra m
}.
\label{eq:SM_five_S1_signed_permutation}
\end{equation}

Its matrix elements are

\begin{equation}
\bra n\widetilde S_1\ket m
=
(-1)^{b_2(m)+b_3(m)}
\delta_{n,m\oplus18}.
\label{eq:SM_five_S1_matrix_elements}
\end{equation}

For example,

\begin{align}
\widetilde S_1\ket0
&=
\ket{18},
\nonumber\\
\widetilde S_1\ket{18}
&=
\ket0,
\nonumber\\
\widetilde S_1\ket9
&=
-\ket{27},
\nonumber\\
\widetilde S_1\ket{27}
&=
-\ket9.
\label{eq:SM_five_S1_examples}
\end{align}

These transformations can be checked directly from the binary states:

\begin{align}
\ket{00000}
&\stackrel{S_1}{\longrightarrow}
\ket{10010},
\nonumber\\
\ket{10010}
&\stackrel{S_1}{\longrightarrow}
\ket{00000},
\nonumber\\
\ket{01001}
&\stackrel{S_1}{\longrightarrow}
-\ket{11011},
\nonumber\\
\ket{11011}
&\stackrel{S_1}{\longrightarrow}
-\ket{01001}.
\end{align}

Thus, the positive and negative amplitudes appearing in
$\ket{\widetilde0_L}$ are paired precisely so that
$\widetilde S_1$ leaves the complete superposition invariant.

Applying the same argument to the remaining generators gives

\begin{align}
\widetilde S_2
&=
\sum_{m=0}^{31}
(-1)^{b_3(m)+b_4(m)}
\ket{m\oplus8\oplus1}\bra m
\nonumber\\
&=
\sum_{m=0}^{31}
(-1)^{b_3(m)+b_4(m)}
\ket{m\oplus9}\bra m,
\label{eq:SM_five_S2_signed_permutation}
\\[1mm]
\widetilde S_3
&=
\sum_{m=0}^{31}
(-1)^{b_4(m)+b_5(m)}
\ket{m\oplus16\oplus4}\bra m
\nonumber\\
&=
\sum_{m=0}^{31}
(-1)^{b_4(m)+b_5(m)}
\ket{m\oplus20}\bra m,
\label{eq:SM_five_S3_signed_permutation}
\\[1mm]
\widetilde S_4
&=
\sum_{m=0}^{31}
(-1)^{b_1(m)+b_5(m)}
\ket{m\oplus8\oplus2}\bra m
\nonumber\\
&=
\sum_{m=0}^{31}
(-1)^{b_1(m)+b_5(m)}
\ket{m\oplus10}\bra m.
\label{eq:SM_five_S4_signed_permutation}
\end{align}

Equations~\eqref{eq:SM_five_S1_signed_permutation}--
\eqref{eq:SM_five_S4_signed_permutation} are the exact
thirty-two-dimensional matrix representations of the embedded
stabilizer generators.

\paragraph{Verification that the transferred operators are stabilizers.}

The transferred generators satisfy all defining properties of
stabilizer generators. First, they are Hermitian,

\begin{equation}
\widetilde S_r^\dagger
=
\widetilde S_r,
\qquad
r=1,\ldots,4.
\label{eq:SM_five_stabilizer_Hermitian}
\end{equation}

They are also involutions,

\begin{align}
\widetilde S_r^2
&=
\Phi S_r\Phi^\dagger
\Phi S_r\Phi^\dagger
\nonumber\\
&=
\Phi S_r^2\Phi^\dagger
\nonumber\\
&=
I_{32}.
\label{eq:SM_five_stabilizer_involution}
\end{align}

Their eigenvalues are therefore restricted to $\pm1$. Moreover,

\begin{align}
[\widetilde S_r,\widetilde S_s]
&=
\Phi[S_r,S_s]\Phi^\dagger
\nonumber\\
&=
0,
\qquad
r,s=1,\ldots,4,
\label{eq:SM_five_stabilizer_commutation}
\end{align}

because the original five-qubit stabilizer generators commute.

Most importantly, the embedded logical states are simultaneous $+1$
eigenstates. Since

\begin{equation}
S_r\ket{\psi_L}
=
\ket{\psi_L},
\end{equation}

we obtain

\begin{align}
\widetilde S_r\ket{\widetilde\psi_L}
&=
\Phi S_r\Phi^\dagger\Phi\ket{\psi_L}
\nonumber\\
&=
\Phi S_r\ket{\psi_L}
\nonumber\\
&=
\Phi\ket{\psi_L}
\nonumber\\
&=
\ket{\widetilde\psi_L}.
\label{eq:SM_five_embedded_stabilization}
\end{align}

In particular,

\begin{equation}
\widetilde S_r\ket{\widetilde0_L}
=
\ket{\widetilde0_L},
\qquad
\widetilde S_r\ket{\widetilde1_L}
=
\ket{\widetilde1_L},
\qquad
r=1,\ldots,4.
\label{eq:SM_five_logical_states_stabilized}
\end{equation}

Therefore, for

\begin{equation}
\ket{\widetilde\psi_L}
=
\alpha\ket{\widetilde0_L}
+
\beta\ket{\widetilde1_L},
\end{equation}

one has

\begin{equation}
\widetilde S_r\ket{\widetilde\psi_L}
=
\ket{\widetilde\psi_L},
\qquad
r=1,\ldots,4.
\end{equation}

The joint $+1$ projector of the four embedded stabilizers is

\begin{equation}
\widetilde P_{\mathcal C}
=
\prod_{r=1}^{4}
\frac{I_{32}+\widetilde S_r}{2}.
\label{eq:SM_five_joint_stabilizer_projector}
\end{equation}

Using $\widetilde S_r=\Phi S_r\Phi^\dagger$, this becomes

\begin{align}
\widetilde P_{\mathcal C}
&=
\Phi
\left[
\prod_{r=1}^{4}
\frac{I+S_r}{2}
\right]
\Phi^\dagger
\nonumber\\
&=
\Phi P_{\mathcal C}\Phi^\dagger.
\end{align}

Since the original common $+1$ eigenspace is

\begin{equation}
\mathcal C
=
\operatorname{span}
\{
\ket{0_L},\ket{1_L}
\},
\end{equation}

the embedded projector is

\begin{equation}
\boxed{
\widetilde P_{\mathcal C}
=
\ket{\widetilde0_L}\bra{\widetilde0_L}
+
\ket{\widetilde1_L}\bra{\widetilde1_L}
}.
\label{eq:SM_five_embedded_code_projector}
\end{equation}

Thus, the simultaneous $+1$ eigenspace of
$\widetilde S_1,\ldots,\widetilde S_4$ is exactly the two-dimensional
embedded logical code space,

\begin{equation}
\widetilde{\mathcal C}
=
\operatorname{span}
\{
\ket{\widetilde0_L},\ket{\widetilde1_L}
\}.
\label{eq:SM_five_common_positive_eigenspace}
\end{equation}

This proves that the transferred operators are genuinely the
stabilizer generators of the embedded code, rather than merely formal
matrix representations of the original qubit operators.

\paragraph{Measurement of the embedded stabilizers.}

In the original realization, syndrome extraction corresponds to
measuring the Pauli products $S_r$. In the thirty-two-level
realization, one measures the transferred two-outcome observables
$\widetilde S_r$. The projectors associated with outcomes $\pm1$ are

\begin{equation}
\widetilde\Pi_r^{(\pm)}
=
\frac{I_{32}\pm\widetilde S_r}{2}.
\label{eq:SM_five_stabilizer_measurement_projectors}
\end{equation}

For a multilevel density operator $\widetilde\rho$, the probability of
outcome $\pm1$ is

\begin{equation}
p_r(\pm)
=
\operatorname{Tr}
\left[
\widetilde\Pi_r^{(\pm)}
\widetilde\rho
\right].
\label{eq:SM_five_stabilizer_probabilities}
\end{equation}

For every embedded logical state $\widetilde\rho_L$,

\begin{equation}
\widetilde P_{\mathcal C}
\widetilde\rho_L
\widetilde P_{\mathcal C}
=
\widetilde\rho_L,
\end{equation}

and hence

\begin{equation}
p_r(+)=1,
\qquad
r=1,\ldots,4.
\end{equation}

The four measurement results therefore yield the no-error syndrome
$0000$ for every state in the embedded logical code space.

\paragraph{Embedded error operators and syndrome preservation.}

For every correctable error $E_a\in\mathcal E_5$, define

\begin{equation}
\widetilde E_a
=
\Phi E_a\Phi^\dagger.
\label{eq:SM_five_transferred_error}
\end{equation}

The embedded error acts on the logical state according to

\begin{align}
\widetilde E_a
\ket{\widetilde\psi_L}
&=
\Phi E_a\Phi^\dagger\Phi\ket{\psi_L}
\nonumber\\
&=
\Phi E_a\ket{\psi_L}.
\label{eq:SM_five_error_intertwining}
\end{align}

Thus, the corrupted multilevel state is exactly the embedded image of
the corresponding corrupted five-qubit state.

Let the syndrome bits of $E_a$ be defined by

\begin{equation}
S_rE_a
=
(-1)^{s_r(a)}
E_aS_r,
\qquad
r=1,\ldots,4.
\label{eq:SM_five_original_syndrome}
\end{equation}

Then

\begin{align}
\widetilde S_r\widetilde E_a
&=
\Phi S_r\Phi^\dagger
\Phi E_a\Phi^\dagger
\nonumber\\
&=
\Phi S_rE_a\Phi^\dagger
\nonumber\\
&=
(-1)^{s_r(a)}
\Phi E_aS_r\Phi^\dagger
\nonumber\\
&=
(-1)^{s_r(a)}
\widetilde E_a\widetilde S_r.
\label{eq:SM_five_embedded_syndrome_relation}
\end{align}

Therefore,

\begin{equation}
\boxed{
\mathbf s(\widetilde E_a)
=
\mathbf s(E_a)
}.
\label{eq:SM_five_syndrome_preservation}
\end{equation}

The four commuting stabilizer measurements generate $2^4=16$ syndrome
sectors. These correspond exactly to the identity branch and the
fifteen single-qubit Pauli errors.

\begin{table}[h]
\caption{Syndromes of the $[[5,1,3]]$ code. The same syndrome labels
apply to the transferred errors and stabilizers in the thirty-two-level
realization.}
\label{tab:SM_five_syndromes}
\begin{ruledtabular}
\begin{tabular}{cc@{\qquad}cc}
Error & Syndrome & Error & Syndrome\\
\hline
$I$   & $0000$ & $X_1$ & $0001$\\
$Y_1$ & $1011$ & $Z_1$ & $1010$\\
$X_2$ & $1000$ & $Y_2$ & $1101$\\
$Z_2$ & $0101$ & $X_3$ & $1100$\\
$Y_3$ & $1110$ & $Z_3$ & $0010$\\
$X_4$ & $0110$ & $Y_4$ & $1111$\\
$Z_4$ & $1001$ & $X_5$ & $0011$\\
$Y_5$ & $0111$ & $Z_5$ & $0100$
\end{tabular}
\end{ruledtabular}
\end{table}

\paragraph{Explicit syndrome example.}

Consider the error $X_1$. From the stabilizers in
Eq.~\eqref{eq:SM_five_stabilizers}, $X_1$ commutes with
$S_1,S_2,$ and $S_3$, but anticommutes with $S_4$. Therefore,

\begin{equation}
\mathbf s(X_1)=0001.
\end{equation}

In the embedded representation,

\begin{align}
\widetilde S_1\widetilde X_1
&=
+\widetilde X_1\widetilde S_1,
\nonumber\\
\widetilde S_2\widetilde X_1
&=
+\widetilde X_1\widetilde S_2,
\nonumber\\
\widetilde S_3\widetilde X_1
&=
+\widetilde X_1\widetilde S_3,
\nonumber\\
\widetilde S_4\widetilde X_1
&=
-\widetilde X_1\widetilde S_4.
\label{eq:SM_five_X1_embedded_commutation}
\end{align}

Acting on an arbitrary embedded logical state gives

\begin{align}
\widetilde S_1
\widetilde X_1
\ket{\widetilde\psi_L}
&=
+\widetilde X_1
\ket{\widetilde\psi_L},
\nonumber\\
\widetilde S_2
\widetilde X_1
\ket{\widetilde\psi_L}
&=
+\widetilde X_1
\ket{\widetilde\psi_L},
\nonumber\\
\widetilde S_3
\widetilde X_1
\ket{\widetilde\psi_L}
&=
+\widetilde X_1
\ket{\widetilde\psi_L},
\nonumber\\
\widetilde S_4
\widetilde X_1
\ket{\widetilde\psi_L}
&=
-\widetilde X_1
\ket{\widetilde\psi_L}.
\end{align}

Hence, measurement of the four embedded stabilizers produces

\begin{equation}
(+1,+1,+1,-1),
\end{equation}

or equivalently the syndrome

\begin{equation}
\mathbf s(\widetilde X_1)=0001.
\end{equation}

This is identical to the syndrome of $X_1$ in the original five-qubit
code.

More generally, the syndrome projector associated with
$\mathbf s=(s_1,s_2,s_3,s_4)$ is

\begin{equation}
\widetilde P_{\mathbf s}
=
\prod_{r=1}^{4}
\frac{
I_{32}
+
(-1)^{s_r}\widetilde S_r
}{2}.
\label{eq:SM_five_syndrome_projector}
\end{equation}

For an error $E_a$ with syndrome $\mathbf s(a)$,

\begin{equation}
\widetilde P_{\mathbf s(a)}
\widetilde E_a
\ket{\widetilde\psi_L}
=
\widetilde E_a
\ket{\widetilde\psi_L},
\end{equation}

while

\begin{equation}
\widetilde P_{\mathbf t}
\widetilde E_a
\ket{\widetilde\psi_L}
=
0,
\qquad
\mathbf t\neq\mathbf s(a).
\end{equation}

Therefore, the sixteen simultaneous eigenspaces of the transferred
stabilizers reproduce exactly the sixteen error sectors of the
five-qubit perfect code.

\paragraph{Embedded recovery operations.}

Let $R_{\mathbf s}$ denote the recovery associated with syndrome
$\mathbf s$. Its embedded representative is

\begin{equation}
\widetilde R_{\mathbf s}
=
\Phi R_{\mathbf s}\Phi^\dagger.
\label{eq:SM_five_embedded_recovery}
\end{equation}

For the nondegenerate five-qubit code, every nontrivial syndrome
identifies a unique single-qubit Pauli error $E_a$. One may therefore
choose

\begin{equation}
R_{\mathbf s(a)}
=
E_a^\dagger.
\end{equation}

Since Pauli operators are Hermitian and unitary,

\begin{equation}
E_a^\dagger
=
E_a,
\qquad
E_a^2=I,
\end{equation}

and hence

\begin{equation}
\widetilde R_{\mathbf s(a)}
=
\widetilde E_a.
\label{eq:SM_five_embedded_recovery_operator}
\end{equation}

For an arbitrary embedded logical state,

\begin{align}
\widetilde R_{\mathbf s(a)}
\widetilde E_a
\ket{\widetilde\psi_L}
&=
\Phi E_a^\dagger\Phi^\dagger
\Phi E_a\Phi^\dagger
\Phi\ket{\psi_L}
\nonumber\\
&=
\Phi E_a^\dagger E_a\ket{\psi_L}
\nonumber\\
&=
\Phi\ket{\psi_L}
\nonumber\\
&=
\ket{\widetilde\psi_L}.
\label{eq:SM_five_exact_state_recovery}
\end{align}

For a logical density operator

\begin{equation}
\widetilde\rho_L
=
\Phi\rho_L\Phi^\dagger,
\end{equation}

the same result is

\begin{align}
\widetilde R_{\mathbf s(a)}
\widetilde E_a
\widetilde\rho_L
\widetilde E_a^\dagger
\widetilde R_{\mathbf s(a)}^\dagger
&=
\widetilde\rho_L.
\label{eq:SM_five_exact_recovery}
\end{align}

As a concrete example, if the measured syndrome is $0001$, the
identified error is $\widetilde X_1$. The corresponding recovery is

\begin{equation}
\widetilde R_{0001}
=
\widetilde X_1.
\end{equation}

Thus,

\begin{align}
\widetilde R_{0001}
\widetilde X_1
\ket{\widetilde\psi_L}
&=
\widetilde X_1^2
\ket{\widetilde\psi_L}
\nonumber\\
&=
\ket{\widetilde\psi_L}.
\end{align}

For the trivial syndrome $0000$, no correction is required, and one
may take

\begin{equation}
\widetilde R_{0000}
=
I_{32}.
\end{equation}

\paragraph{Conclusion of the five-qubit example.}

The thirty-two-level realization preserves the complete
error-correction structure of the $[[5,1,3]]$ perfect code. In
particular,

\begin{enumerate}
\item the logical codewords are mapped isometrically to
      $\ket{\widetilde0_L}$ and $\ket{\widetilde1_L}$;

\item the transferred generators
      $\widetilde S_1,\ldots,\widetilde S_4$ are commuting Hermitian
      involutions whose common $+1$ eigenspace is exactly
      $\widetilde{\mathcal C}$;

\item every single-qubit Pauli error is mapped to an explicit
      multilevel permutation, phase operator, or signed permutation;

\item all sixteen syndrome assignments are preserved exactly; and

\item the transferred recovery operator restores the original embedded
      logical state.
\end{enumerate}

The embedded operators therefore reproduce the full logical,
stabilizer, syndrome, and recovery structure of the five-qubit perfect
code within a single thirty-two-level system.

This conclusion applies to the transferred error family
$\widetilde{\mathcal E}_5$. It does not by itself imply correction of
every physical operator acting on an arbitrary multilevel extension.
If the physical Hilbert space has dimension $D>32$, errors that couple
the embedded sector to additional leakage levels must be analyzed
separately using the corresponding Knill--Laflamme conditions.
\subsection{S4.3 Seven-qubit Steane code}
\label{subsec:SM_steane}

The $[[7,1,3]]$ Steane code encodes one logical qubit into seven
physical qubits and corrects an arbitrary error acting on any one
physical qubit. It is a Calderbank--Shor--Steane (CSS) stabilizer code
derived from the classical $[7,4,3]$ Hamming code
\cite{Steane1996}. It therefore provides a structurally distinct
illustration of the embedding construction. In contrast to the
five-qubit perfect code, the Steane code separates its stabilizer
generators into independent $X$-type and $Z$-type parity checks. The
example consequently demonstrates that the embedding preserves the CSS
structure, six-bit syndrome assignments, and recovery of arbitrary
single-qubit Pauli errors in a single one-hundred-and-twenty-eight-level
system.

A convenient set of stabilizer generators is

\begin{align}
S_1&=IIIXXXX, &
S_4&=IIIZZZZ,
\nonumber\\
S_2&=IXXIIXX, &
S_5&=IZZIIZZ,
\nonumber\\
S_3&=XIXIXIX, &
S_6&=ZIZIZIZ.
\label{eq:SM_steane_stabilizers}
\end{align}

Explicitly,

\begin{align}
S_1&=
I_1I_2I_3X_4X_5X_6X_7,
&
S_4&=
I_1I_2I_3Z_4Z_5Z_6Z_7,
\nonumber\\
S_2&=
I_1X_2X_3I_4I_5X_6X_7,
&
S_5&=
I_1Z_2Z_3I_4I_5Z_6Z_7,
\nonumber\\
S_3&=
X_1I_2X_3I_4X_5I_6X_7,
&
S_6&=
Z_1I_2Z_3I_4Z_5I_6Z_7.
\label{eq:SM_steane_stabilizers_explicit}
\end{align}

The first three generators are $X$-type stabilizers, whereas the final
three are the corresponding $Z$-type stabilizers. A convenient choice
of logical Pauli operators is

\begin{equation}
\overline X
=
X^{\otimes7},
\qquad
\overline Z
=
Z^{\otimes7}.
\label{eq:SM_steane_logicals}
\end{equation}

Let

\begin{align}
\mathcal C^\perp
=
\{&
0000000,\,
0001111,\,
0110011,\,
0111100,
\nonumber\\
&
1010101,\,
1011010,\,
1100110,\,
1101001
\}.
\label{eq:SM_steane_dual_code}
\end{align}

The logical basis states are

\begin{align}
\ket{0_L}
&=
\frac{1}{\sqrt8}
\sum_{c\in\mathcal C^\perp}
\ket c,
\nonumber\\
\ket{1_L}
&=
\overline X\ket{0_L}
=
X^{\otimes7}\ket{0_L}.
\label{eq:SM_steane_codewords}
\end{align}

Equivalently,

\begin{align}
\ket{0_L}
=
\frac{1}{\sqrt8}
\big(&
\ket{0000000}
+\ket{0001111}
+\ket{0110011}
+\ket{0111100}
\nonumber\\
&
+\ket{1010101}
+\ket{1011010}
+\ket{1100110}
+\ket{1101001}
\big),
\label{eq:SM_steane_zero_explicit}
\end{align}

and

\begin{align}
\ket{1_L}
=
\frac{1}{\sqrt8}
\big(&
\ket{1111111}
+\ket{1110000}
+\ket{1001100}
+\ket{1000011}
\nonumber\\
&
+\ket{0101010}
+\ket{0100101}
+\ket{0011001}
+\ket{0010110}
\big).
\label{eq:SM_steane_one_explicit}
\end{align}

The correctable single-qubit Pauli family is

\begin{equation}
\mathcal E_7
=
\{I\}
\cup
\{X_j,Y_j,Z_j:j=1,\ldots,7\}.
\label{eq:SM_steane_error_set}
\end{equation}

This family contains the identity and the twenty-one nontrivial
single-qubit Pauli errors.

\paragraph{Embedding into a one-hundred-and-twenty-eight-level system.}

For $D=128$, the embedding map is

\begin{equation}
\Phi
=
\sum_{m=0}^{127}
\ket m\bra{b(m)},
\label{eq:SM_steane_phi}
\end{equation}

where $\ket{b(m)}$ denotes the seven-qubit computational basis state
whose binary representation is the integer $m$. Explicitly,

\begin{equation}
\ket{b_1b_2b_3b_4b_5b_6b_7}
\longmapsto
\ket{
64b_1+32b_2+16b_3+8b_4+4b_5+2b_6+b_7
}.
\label{eq:SM_steane_binary_decimal}
\end{equation}

For example,

\begin{align}
\ket{0000000}&\longmapsto\ket0,
&
\ket{0001111}&\longmapsto\ket{15},
\nonumber\\
\ket{0110011}&\longmapsto\ket{51},
&
\ket{0111100}&\longmapsto\ket{60},
\nonumber\\
\ket{1010101}&\longmapsto\ket{85},
&
\ket{1101001}&\longmapsto\ket{105}.
\end{align}

Using the ordered computational basis

\begin{equation}
\mathcal B_q
=
\{
\ket{0000000},
\ket{0000001},
\ldots,
\ket{1111111}
\},
\end{equation}

and the ordered multilevel basis

\begin{equation}
\mathcal B_{128}
=
\{
\ket0,\ket1,\ldots,\ket{127}
\},
\end{equation}

the matrix representation of $\Phi$ is numerically $I_{128}$. Its
domain and codomain are nevertheless physically distinct: the columns
represent states of seven physical qubits, whereas the rows represent
the energy levels of one one-hundred-and-twenty-eight-dimensional
system.

Because the complete seven-qubit Hilbert space is embedded into a
physical space of the same dimension,

\begin{equation}
\Phi^\dagger\Phi
=
I_{\mathcal H_q},
\qquad
\Phi\Phi^\dagger
=
I_{128}.
\label{eq:SM_steane_isometry}
\end{equation}

\paragraph{Embedded logical codewords.}

Applying $\Phi$ term by term to
Eq.~\eqref{eq:SM_steane_zero_explicit} gives

\begin{align}
\ket{\widetilde0_L}
&=
\Phi\ket{0_L}
\nonumber\\
&=
\frac{1}{\sqrt8}
\big(
\ket0
+\ket{15}
+\ket{51}
+\ket{60}
+\ket{85}
+\ket{90}
+\ket{102}
+\ket{105}
\big).
\label{eq:SM_steane_embedded_zero}
\end{align}

The transferred logical-$X$ operator is

\begin{equation}
\widetilde{\overline X}
=
\Phi\overline X\Phi^\dagger.
\label{eq:SM_steane_embedded_logical_X_definition}
\end{equation}

Because $\overline X=X^{\otimes7}$ complements all seven binary digits,

\begin{equation}
\widetilde{\overline X}
=
\sum_{m=0}^{127}
\ket{m\oplus127}\bra m,
\label{eq:SM_steane_embedded_logical_X}
\end{equation}

where $127=(1111111)_2$ and $\oplus$ denotes bitwise XOR. Therefore,

\begin{align}
\ket{\widetilde1_L}
&=
\Phi\ket{1_L}
=
\widetilde{\overline X}
\ket{\widetilde0_L}
\nonumber\\
&=
\frac{1}{\sqrt8}
\big(
\ket{127}
+\ket{112}
+\ket{76}
+\ket{67}
+\ket{42}
+\ket{37}
+\ket{25}
+\ket{22}
\big).
\label{eq:SM_steane_embedded_one}
\end{align}

The transferred logical-$Z$ operator is

\begin{align}
\widetilde{\overline Z}
&=
\Phi\overline Z\Phi^\dagger
\nonumber\\
&=
\sum_{m=0}^{127}
(-1)^{
b_1(m)+b_2(m)+b_3(m)+b_4(m)+
b_5(m)+b_6(m)+b_7(m)
}
\ket m\bra m.
\label{eq:SM_steane_embedded_logical_Z}
\end{align}

Every binary string appearing in $\ket{0_L}$ has even Hamming weight,
whereas every binary string appearing in $\ket{1_L}$ has odd Hamming
weight. Consequently,

\begin{align}
\widetilde{\overline Z}
\ket{\widetilde0_L}
&=
\ket{\widetilde0_L},
\nonumber\\
\widetilde{\overline Z}
\ket{\widetilde1_L}
&=
-\ket{\widetilde1_L},
\label{eq:SM_steane_logical_Z_action}
\end{align}

and

\begin{align}
\widetilde{\overline X}
\ket{\widetilde0_L}
&=
\ket{\widetilde1_L},
\nonumber\\
\widetilde{\overline X}
\ket{\widetilde1_L}
&=
\ket{\widetilde0_L}.
\label{eq:SM_steane_logical_X_action}
\end{align}

An arbitrary logical state

\begin{equation}
\ket{\psi_L}
=
\alpha\ket{0_L}
+
\beta\ket{1_L}
\end{equation}

is therefore represented as

\begin{equation}
\ket{\widetilde\psi_L}
=
\Phi\ket{\psi_L}
=
\alpha\ket{\widetilde0_L}
+
\beta\ket{\widetilde1_L},
\label{eq:SM_steane_general_embedded_state}
\end{equation}

with the logical amplitudes $\alpha$ and $\beta$ unchanged.

\paragraph{Transferred single-qubit Pauli operators.}

Before constructing the embedded stabilizers, we obtain the exact
multilevel representation of the virtual single-qubit Pauli
operators. Let

\begin{equation}
w_j
=
2^{7-j},
\qquad
(w_1,w_2,w_3,w_4,w_5,w_6,w_7)
=
(64,32,16,8,4,2,1),
\label{eq:SM_steane_binary_weights}
\end{equation}

and let $b_j(m)\in\{0,1\}$ denote the $j$th binary digit of $m$, with
$b_1$ the most significant digit.

A bit flip on virtual qubit $j$ complements only the corresponding
binary digit. Hence,

\begin{equation}
\widetilde X_j
=
\Phi X_j\Phi^\dagger
=
\sum_{m=0}^{127}
\ket{m\oplus w_j}\bra m.
\label{eq:SM_steane_embedded_X}
\end{equation}

Its matrix elements in the multilevel basis are

\begin{equation}
\bra n\widetilde X_j\ket m
=
\delta_{n,m\oplus w_j}.
\label{eq:SM_steane_X_matrix_elements}
\end{equation}

Thus, $\widetilde X_j$ is an exact permutation matrix. For example,

\begin{align}
\widetilde X_1
={}&
\ket{64}\bra0
+\ket{65}\bra1
+\cdots
+\ket{127}\bra{63}
\nonumber\\
&+
\ket0\bra{64}
+\ket1\bra{65}
+\cdots
+\ket{63}\bra{127}.
\label{eq:SM_steane_X1_explicit}
\end{align}

Accordingly,

\begin{equation}
\ket0
\stackrel{\widetilde X_1}{\longleftrightarrow}
\ket{64},
\qquad
\ket{15}
\stackrel{\widetilde X_1}{\longleftrightarrow}
\ket{79},
\end{equation}

which are the embedded images of

\begin{equation}
\ket{0000000}
\stackrel{X_1}{\longleftrightarrow}
\ket{1000000},
\qquad
\ket{0001111}
\stackrel{X_1}{\longleftrightarrow}
\ket{1001111}.
\end{equation}

Similarly,

\begin{equation}
\ket{15}
\stackrel{\widetilde X_4}{\longleftrightarrow}
\ket7,
\end{equation}

because $w_4=8$ and $15\oplus8=7$.

A phase flip on virtual qubit $j$ produces the phase
$(-1)^{b_j(m)}$. Its embedded representation is therefore

\begin{equation}
\widetilde Z_j
=
\Phi Z_j\Phi^\dagger
=
\sum_{m=0}^{127}
(-1)^{b_j(m)}
\ket m\bra m,
\label{eq:SM_steane_embedded_Z}
\end{equation}

with matrix elements

\begin{equation}
\bra n\widetilde Z_j\ket m
=
(-1)^{b_j(m)}
\delta_{nm}.
\label{eq:SM_steane_Z_matrix_elements}
\end{equation}

Finally,

\begin{equation}
\widetilde Y_j
=
\Phi Y_j\Phi^\dagger
=
i\widetilde X_j\widetilde Z_j,
\label{eq:SM_steane_embedded_Y}
\end{equation}

so that

\begin{equation}
\widetilde Y_j
=
i
\sum_{m=0}^{127}
(-1)^{b_j(m)}
\ket{m\oplus w_j}\bra m.
\label{eq:SM_steane_Y_explicit}
\end{equation}

Thus,

\[
\widetilde X_j
\text{ is a level permutation},\qquad
\widetilde Z_j
\text{ is diagonal},\qquad
\widetilde Y_j
\text{ is a signed level permutation}.
\]

Because the transfer is an isometric conjugation, the complete
single-qubit Pauli algebra is preserved:

\begin{align}
\widetilde X_j^2
&=
\widetilde Y_j^2
=
\widetilde Z_j^2
=
I_{128},
\nonumber\\
\widetilde X_j\widetilde Z_j
&=
-\widetilde Z_j\widetilde X_j,
\nonumber\\
[\widetilde P_j,\widetilde Q_k]
&=
0,
\qquad
j\neq k,
\label{eq:SM_steane_transferred_Pauli_algebra}
\end{align}

where $P,Q\in\{X,Y,Z\}$.

\paragraph{Construction of the embedded stabilizer generators.}

Each stabilizer generator is transferred according to

\begin{equation}
\widetilde S_r
=
\Phi S_r\Phi^\dagger,
\qquad
r=1,\ldots,6.
\label{eq:SM_steane_embedded_stabilizer_definition}
\end{equation}

Using $\Phi^\dagger\Phi=I_{\mathcal H_q}$ between successive factors,
the embedded generators are

\begin{align}
\widetilde S_1
&=
\widetilde X_4
\widetilde X_5
\widetilde X_6
\widetilde X_7,
&
\widetilde S_4
&=
\widetilde Z_4
\widetilde Z_5
\widetilde Z_6
\widetilde Z_7,
\nonumber\\
\widetilde S_2
&=
\widetilde X_2
\widetilde X_3
\widetilde X_6
\widetilde X_7,
&
\widetilde S_5
&=
\widetilde Z_2
\widetilde Z_3
\widetilde Z_6
\widetilde Z_7,
\nonumber\\
\widetilde S_3
&=
\widetilde X_1
\widetilde X_3
\widetilde X_5
\widetilde X_7,
&
\widetilde S_6
&=
\widetilde Z_1
\widetilde Z_3
\widetilde Z_5
\widetilde Z_7.
\label{eq:SM_steane_embedded_stabilizers}
\end{align}

The CSS structure is directly visible in
Eq.~\eqref{eq:SM_steane_embedded_stabilizers}: the first three
generators are pure level permutations, whereas the last three are
diagonal binary-phase operators.

We first derive $\widetilde S_1$ explicitly. Acting on a computational
basis state gives

\begin{align}
S_1
\ket{b_1b_2b_3b_4b_5b_6b_7}
=
\ket{
b_1,\,
b_2,\,
b_3,\,
b_4\oplus1,\,
b_5\oplus1,\,
b_6\oplus1,\,
b_7\oplus1
}.
\label{eq:SM_steane_S1_computational_action}
\end{align}

The fourth, fifth, sixth, and seventh digits have binary weights
$8,4,2,$ and $1$. Therefore,

\begin{equation}
\widetilde S_1\ket m
=
\ket{m\oplus8\oplus4\oplus2\oplus1}
=
\ket{m\oplus15}.
\label{eq:SM_steane_S1_level_action}
\end{equation}

Hence, the exact multilevel representation is

\begin{equation}
\boxed{
\widetilde S_1
=
\sum_{m=0}^{127}
\ket{m\oplus15}\bra m
}.
\label{eq:SM_steane_S1_permutation}
\end{equation}

Its matrix elements are

\begin{equation}
\bra n\widetilde S_1\ket m
=
\delta_{n,m\oplus15}.
\label{eq:SM_steane_S1_matrix_elements}
\end{equation}

For example,

\begin{align}
\widetilde S_1\ket0
&=
\ket{15},
&
\widetilde S_1\ket{15}
&=
\ket0,
\nonumber\\
\widetilde S_1\ket{51}
&=
\ket{60},
&
\widetilde S_1\ket{60}
&=
\ket{51},
\nonumber\\
\widetilde S_1\ket{85}
&=
\ket{90},
&
\widetilde S_1\ket{90}
&=
\ket{85},
\nonumber\\
\widetilde S_1\ket{102}
&=
\ket{105},
&
\widetilde S_1\ket{105}
&=
\ket{102}.
\label{eq:SM_steane_S1_examples}
\end{align}

These pairings show directly that

\begin{equation}
\widetilde S_1
\ket{\widetilde0_L}
=
\ket{\widetilde0_L}.
\end{equation}

The other two $X$-type generators are constructed identically. Their
binary masks are

\begin{align}
w_2\oplus w_3\oplus w_6\oplus w_7
&=
32\oplus16\oplus2\oplus1
=
51,
\nonumber\\
w_1\oplus w_3\oplus w_5\oplus w_7
&=
64\oplus16\oplus4\oplus1
=
85.
\end{align}

Therefore,

\begin{equation}
\boxed{
\widetilde S_2
=
\sum_{m=0}^{127}
\ket{m\oplus51}\bra m
},
\label{eq:SM_steane_S2_permutation}
\end{equation}

and

\begin{equation}
\boxed{
\widetilde S_3
=
\sum_{m=0}^{127}
\ket{m\oplus85}\bra m
}.
\label{eq:SM_steane_S3_permutation}
\end{equation}

We next derive the first $Z$-type stabilizer. Acting on a computational
basis state gives

\begin{align}
S_4
\ket{b_1b_2b_3b_4b_5b_6b_7}
=
(-1)^{b_4+b_5+b_6+b_7}
\ket{b_1b_2b_3b_4b_5b_6b_7}.
\label{eq:SM_steane_S4_computational_action}
\end{align}

After embedding,

\begin{equation}
\widetilde S_4\ket m
=
(-1)^{
b_4(m)+b_5(m)+b_6(m)+b_7(m)
}
\ket m.
\label{eq:SM_steane_S4_level_action}
\end{equation}

Thus,

\begin{equation}
\boxed{
\widetilde S_4
=
\sum_{m=0}^{127}
(-1)^{
b_4(m)+b_5(m)+b_6(m)+b_7(m)
}
\ket m\bra m
}.
\label{eq:SM_steane_S4_diagonal}
\end{equation}

Its matrix elements are

\begin{equation}
\bra n\widetilde S_4\ket m
=
(-1)^{
b_4(m)+b_5(m)+b_6(m)+b_7(m)
}
\delta_{nm}.
\label{eq:SM_steane_S4_matrix_elements}
\end{equation}

Every level appearing in $\ket{\widetilde0_L}$ contains an even number
of ones among the fourth through seventh binary digits. Hence,

\begin{equation}
\widetilde S_4
\ket{\widetilde0_L}
=
\ket{\widetilde0_L}.
\end{equation}

The remaining $Z$-type generators are

\begin{equation}
\boxed{
\widetilde S_5
=
\sum_{m=0}^{127}
(-1)^{
b_2(m)+b_3(m)+b_6(m)+b_7(m)
}
\ket m\bra m
},
\label{eq:SM_steane_S5_diagonal}
\end{equation}

and

\begin{equation}
\boxed{
\widetilde S_6
=
\sum_{m=0}^{127}
(-1)^{
b_1(m)+b_3(m)+b_5(m)+b_7(m)
}
\ket m\bra m
}.
\label{eq:SM_steane_S6_diagonal}
\end{equation}

Equations~\eqref{eq:SM_steane_S1_permutation}--
\eqref{eq:SM_steane_S6_diagonal} are the exact
one-hundred-and-twenty-eight-dimensional representations of the six
embedded stabilizer generators.

\paragraph{Verification that the transferred operators are stabilizers.}

The transferred generators satisfy all defining properties of
stabilizer generators. First, they are Hermitian:

\begin{equation}
\widetilde S_r^\dagger
=
\widetilde S_r,
\qquad
r=1,\ldots,6.
\label{eq:SM_steane_stabilizer_Hermitian}
\end{equation}

They are also involutions:

\begin{align}
\widetilde S_r^2
&=
\Phi S_r\Phi^\dagger
\Phi S_r\Phi^\dagger
\nonumber\\
&=
\Phi S_r^2\Phi^\dagger
\nonumber\\
&=
I_{128}.
\label{eq:SM_steane_stabilizer_involution}
\end{align}

Their eigenvalues are therefore restricted to $\pm1$. Moreover,

\begin{align}
[\widetilde S_r,\widetilde S_s]
&=
\Phi[S_r,S_s]\Phi^\dagger
\nonumber\\
&=
0,
\qquad
r,s=1,\ldots,6,
\label{eq:SM_steane_stabilizer_commutation}
\end{align}

because the original Steane-code stabilizers commute.

The commutation of the transferred $X$- and $Z$-type generators can
also be understood directly from the CSS structure. Every $X$-type
generator overlaps with every $Z$-type generator on an even number of
virtual qubits. Each overlapping $X$--$Z$ pair contributes one minus
sign, so the total sign is positive.

Most importantly, the embedded logical states are simultaneous $+1$
eigenstates. Since

\begin{equation}
S_r\ket{\psi_L}
=
\ket{\psi_L},
\qquad
r=1,\ldots,6,
\end{equation}

we obtain

\begin{align}
\widetilde S_r
\ket{\widetilde\psi_L}
&=
\Phi S_r\Phi^\dagger
\Phi\ket{\psi_L}
\nonumber\\
&=
\Phi S_r\ket{\psi_L}
\nonumber\\
&=
\Phi\ket{\psi_L}
\nonumber\\
&=
\ket{\widetilde\psi_L}.
\label{eq:SM_steane_embedded_stabilization}
\end{align}

In particular,

\begin{equation}
\widetilde S_r\ket{\widetilde0_L}
=
\ket{\widetilde0_L},
\qquad
\widetilde S_r\ket{\widetilde1_L}
=
\ket{\widetilde1_L},
\qquad
r=1,\ldots,6.
\label{eq:SM_steane_logical_states_stabilized}
\end{equation}

The joint $+1$ projector of the six embedded stabilizers is

\begin{equation}
\widetilde P_{\mathcal C}
=
\prod_{r=1}^{6}
\frac{I_{128}+\widetilde S_r}{2}.
\label{eq:SM_steane_joint_stabilizer_projector}
\end{equation}

Using $\widetilde S_r=\Phi S_r\Phi^\dagger$, this becomes

\begin{align}
\widetilde P_{\mathcal C}
&=
\Phi
\left[
\prod_{r=1}^{6}
\frac{I+S_r}{2}
\right]
\Phi^\dagger
\nonumber\\
&=
\Phi P_{\mathcal C}\Phi^\dagger.
\label{eq:SM_steane_projector_transfer}
\end{align}

Since the common $+1$ eigenspace of the original six independent
stabilizers is

\begin{equation}
\mathcal C
=
\operatorname{span}
\{
\ket{0_L},\ket{1_L}
\},
\end{equation}

the embedded projector is

\begin{equation}
\boxed{
\widetilde P_{\mathcal C}
=
\ket{\widetilde0_L}\bra{\widetilde0_L}
+
\ket{\widetilde1_L}\bra{\widetilde1_L}
}.
\label{eq:SM_steane_embedded_code_projector}
\end{equation}

Thus, the simultaneous $+1$ eigenspace of
$\widetilde S_1,\ldots,\widetilde S_6$ is exactly

\begin{equation}
\widetilde{\mathcal C}
=
\operatorname{span}
\{
\ket{\widetilde0_L},\ket{\widetilde1_L}
\}.
\label{eq:SM_steane_common_positive_eigenspace}
\end{equation}

This establishes that the transferred operators are genuine
stabilizer generators of the embedded Steane code and that the original
CSS code space is preserved exactly.

\paragraph{Measurement of the embedded stabilizers.}

In the original seven-qubit realization, syndrome extraction is
performed by measuring the six Pauli products $S_r$. In the
one-hundred-and-twenty-eight-level realization, one instead measures
the transferred two-outcome observables $\widetilde S_r$. The
projectors associated with outcomes $\pm1$ are

\begin{equation}
\widetilde\Pi_r^{(\pm)}
=
\frac{I_{128}\pm\widetilde S_r}{2}.
\label{eq:SM_steane_stabilizer_measurement_projectors}
\end{equation}

For a multilevel density operator $\widetilde\rho$, the probability of
obtaining outcome $\pm1$ is

\begin{equation}
p_r(\pm)
=
\operatorname{Tr}
\left[
\widetilde\Pi_r^{(\pm)}
\widetilde\rho
\right].
\label{eq:SM_steane_stabilizer_probabilities}
\end{equation}

For every embedded logical state $\widetilde\rho_L$,

\begin{equation}
\widetilde P_{\mathcal C}
\widetilde\rho_L
\widetilde P_{\mathcal C}
=
\widetilde\rho_L,
\end{equation}

and therefore

\begin{equation}
p_r(+)=1,
\qquad
r=1,\ldots,6.
\end{equation}

The six measurements consequently produce the no-error syndrome
$000000$ for every state in the embedded logical code space.

The CSS separation is also preserved operationally. The first three
measurements, corresponding to the embedded $X$-type stabilizers,
detect phase errors, whereas the final three measurements,
corresponding to the embedded $Z$-type stabilizers, detect bit-flip
errors.

\paragraph{Embedded error operators and syndrome preservation.}

For every correctable error $E_a\in\mathcal E_7$, define the embedded
operator

\begin{equation}
\widetilde E_a
=
\Phi E_a\Phi^\dagger.
\label{eq:SM_steane_transferred_error}
\end{equation}

Its action on an embedded logical state is

\begin{align}
\widetilde E_a
\ket{\widetilde\psi_L}
&=
\Phi E_a\Phi^\dagger
\Phi\ket{\psi_L}
\nonumber\\
&=
\Phi E_a\ket{\psi_L}.
\label{eq:SM_steane_error_intertwining}
\end{align}

Thus, the corrupted multilevel state is exactly the embedded image of
the corresponding corrupted seven-qubit state.

Let the six syndrome bits of $E_a$ be defined by

\begin{equation}
S_rE_a
=
(-1)^{s_r(a)}
E_aS_r,
\qquad
r=1,\ldots,6.
\label{eq:SM_steane_original_syndrome}
\end{equation}

Then

\begin{align}
\widetilde S_r\widetilde E_a
&=
\Phi S_r\Phi^\dagger
\Phi E_a\Phi^\dagger
\nonumber\\
&=
\Phi S_rE_a\Phi^\dagger
\nonumber\\
&=
(-1)^{s_r(a)}
\Phi E_aS_r\Phi^\dagger
\nonumber\\
&=
(-1)^{s_r(a)}
\widetilde E_a\widetilde S_r.
\label{eq:SM_steane_embedded_syndrome_relation}
\end{align}

Therefore,

\begin{equation}
\boxed{
\mathbf s(\widetilde E_a)
=
\mathbf s(E_a)
}.
\label{eq:SM_steane_syndrome_preservation}
\end{equation}

To display the CSS syndrome structure compactly, define the Hamming
columns

\begin{equation}
\mathbf h_1=001,\quad
\mathbf h_2=010,\quad
\mathbf h_3=011,\quad
\mathbf h_4=100,\quad
\mathbf h_5=101,\quad
\mathbf h_6=110,\quad
\mathbf h_7=111.
\label{eq:SM_steane_hamming_columns}
\end{equation}

With the syndrome ordered as

\begin{equation}
\mathbf s
=
(s_1,s_2,s_3,s_4,s_5,s_6),
\end{equation}

where the first three entries correspond to the $X$-type stabilizers
and the final three entries correspond to the $Z$-type stabilizers, one
has

\begin{align}
\mathbf s(X_j)
&=
(000,\mathbf h_j),
\nonumber\\
\mathbf s(Z_j)
&=
(\mathbf h_j,000),
\nonumber\\
\mathbf s(Y_j)
&=
(\mathbf h_j,\mathbf h_j).
\label{eq:SM_steane_compact_syndromes}
\end{align}

Thus, the $Z$-type stabilizers identify the location of an $X$ error,
the $X$-type stabilizers identify the location of a $Z$ error, and a
$Y$ error activates both syndrome components.

The full correctable-error syndrome table is

\begin{table}[h]
\caption{Six-bit syndromes of the $[[7,1,3]]$ Steane code, ordered as
$(S_1,S_2,S_3,S_4,S_5,S_6)$. The same syndrome assignments apply to
the transferred errors and stabilizers in the one-hundred-and-twenty-
eight-level realization.}
\label{tab:SM_steane_syndromes}
\begin{ruledtabular}
\begin{tabular}{cc@{\qquad}cc}
Error & Syndrome & Error & Syndrome\\
\hline
$I$   & $000000$ & $X_1$ & $000001$\\
$X_2$ & $000010$ & $X_3$ & $000011$\\
$X_4$ & $000100$ & $X_5$ & $000101$\\
$X_6$ & $000110$ & $X_7$ & $000111$\\
$Z_1$ & $001000$ & $Z_2$ & $010000$\\
$Z_3$ & $011000$ & $Z_4$ & $100000$\\
$Z_5$ & $101000$ & $Z_6$ & $110000$\\
$Z_7$ & $111000$ & $Y_1$ & $001001$\\
$Y_2$ & $010010$ & $Y_3$ & $011011$\\
$Y_4$ & $100100$ & $Y_5$ & $101101$\\
$Y_6$ & $110110$ & $Y_7$ & $111111$
\end{tabular}
\end{ruledtabular}
\end{table}

The six stabilizer measurements define $2^6=64$ simultaneous syndrome
sectors. The identity and the twenty-one single-qubit Pauli errors
occupy twenty-two distinct sectors. The remaining sectors correspond
to higher-weight error patterns or to other representatives of the
associated stabilizer cosets.

\paragraph{Explicit syndrome example.}

Consider the error $X_1$. From
Eq.~\eqref{eq:SM_steane_stabilizers}, $X_1$ commutes with the three
$X$-type generators $S_1,S_2,$ and $S_3$. It also commutes with
$S_4$ and $S_5$, whose support does not contain qubit $1$, but
anticommutes with $S_6$, which contains $Z_1$. Hence,

\begin{equation}
\mathbf s(X_1)
=
000001.
\label{eq:SM_steane_X1_syndrome}
\end{equation}

The embedded operators obey the identical relations:

\begin{align}
\widetilde S_1\widetilde X_1
&=
+\widetilde X_1\widetilde S_1,
\nonumber\\
\widetilde S_2\widetilde X_1
&=
+\widetilde X_1\widetilde S_2,
\nonumber\\
\widetilde S_3\widetilde X_1
&=
+\widetilde X_1\widetilde S_3,
\nonumber\\
\widetilde S_4\widetilde X_1
&=
+\widetilde X_1\widetilde S_4,
\nonumber\\
\widetilde S_5\widetilde X_1
&=
+\widetilde X_1\widetilde S_5,
\nonumber\\
\widetilde S_6\widetilde X_1
&=
-\widetilde X_1\widetilde S_6.
\label{eq:SM_steane_X1_embedded_commutation}
\end{align}

Acting on an arbitrary embedded logical state gives

\begin{align}
\widetilde S_r
\widetilde X_1
\ket{\widetilde\psi_L}
&=
+\widetilde X_1
\ket{\widetilde\psi_L},
\qquad
r=1,\ldots,5,
\nonumber\\
\widetilde S_6
\widetilde X_1
\ket{\widetilde\psi_L}
&=
-\widetilde X_1
\ket{\widetilde\psi_L}.
\label{eq:SM_steane_X1_eigenvalues}
\end{align}

Measurement of the six embedded stabilizers therefore produces

\begin{equation}
(+1,+1,+1,+1,+1,-1),
\end{equation}

or equivalently,

\begin{equation}
\mathbf s(\widetilde X_1)
=
000001.
\end{equation}

As a phase-error example, $Z_1$ anticommutes only with $S_3$ among the
six generators. Hence,

\begin{equation}
\mathbf s(Z_1)
=
001000,
\qquad
\mathbf s(\widetilde Z_1)
=
001000.
\end{equation}

Because $Y_1=iX_1Z_1$, its syndrome is the bitwise sum of the $X_1$
and $Z_1$ syndromes:

\begin{equation}
\mathbf s(Y_1)
=
001001,
\qquad
\mathbf s(\widetilde Y_1)
=
001001.
\end{equation}

These examples explicitly demonstrate the preservation of the CSS
separation between bit-flip and phase-flip syndrome information.

More generally, the embedded syndrome projector associated with

\begin{equation}
\mathbf s
=
(s_1,s_2,s_3,s_4,s_5,s_6)
\end{equation}

is

\begin{equation}
\widetilde P_{\mathbf s}
=
\prod_{r=1}^{6}
\frac{
I_{128}
+
(-1)^{s_r}\widetilde S_r
}{2}.
\label{eq:SM_steane_syndrome_projector}
\end{equation}

For an error $E_a$ with syndrome $\mathbf s(a)$,

\begin{equation}
\widetilde P_{\mathbf s(a)}
\widetilde E_a
\ket{\widetilde\psi_L}
=
\widetilde E_a
\ket{\widetilde\psi_L},
\label{eq:SM_steane_correct_syndrome_projection}
\end{equation}

whereas

\begin{equation}
\widetilde P_{\mathbf t}
\widetilde E_a
\ket{\widetilde\psi_L}
=
0,
\qquad
\mathbf t\neq\mathbf s(a).
\label{eq:SM_steane_wrong_syndrome_projection}
\end{equation}

Thus, the simultaneous eigenspaces of the transferred stabilizers
reproduce exactly the syndrome sectors occupied by the correctable
single-qubit errors.

\paragraph{Embedded recovery operations.}

Let $R_{\mathbf s}$ denote the recovery operation associated with
syndrome $\mathbf s$. Its embedded representative is

\begin{equation}
\widetilde R_{\mathbf s}
=
\Phi R_{\mathbf s}\Phi^\dagger.
\label{eq:SM_steane_embedded_recovery}
\end{equation}

For every single-qubit Pauli error $E_a\in\mathcal E_7$, the syndrome
$\mathbf s(a)$ uniquely identifies the error within the correctable
family. One may therefore choose

\begin{equation}
R_{\mathbf s(a)}
=
E_a^\dagger.
\end{equation}

Because Pauli operators are Hermitian and unitary,

\begin{equation}
E_a^\dagger
=
E_a,
\qquad
E_a^2
=
I,
\end{equation}

and hence

\begin{equation}
\widetilde R_{\mathbf s(a)}
=
\widetilde E_a.
\label{eq:SM_steane_embedded_recovery_operator}
\end{equation}

For an arbitrary embedded logical state,

\begin{align}
\widetilde R_{\mathbf s(a)}
\widetilde E_a
\ket{\widetilde\psi_L}
&=
\Phi E_a^\dagger\Phi^\dagger
\Phi E_a\Phi^\dagger
\Phi\ket{\psi_L}
\nonumber\\
&=
\Phi E_a^\dagger E_a
\ket{\psi_L}
\nonumber\\
&=
\Phi\ket{\psi_L}
\nonumber\\
&=
\ket{\widetilde\psi_L}.
\label{eq:SM_steane_exact_state_recovery}
\end{align}

For an embedded logical density operator

\begin{equation}
\widetilde\rho_L
=
\Phi\rho_L\Phi^\dagger,
\end{equation}

the same result is

\begin{align}
\widetilde R_{\mathbf s(a)}
\widetilde E_a
\widetilde\rho_L
\widetilde E_a^\dagger
\widetilde R_{\mathbf s(a)}^\dagger
&=
\widetilde\rho_L.
\label{eq:SM_steane_exact_recovery}
\end{align}

As a concrete example, if the measured syndrome is $000001$, the
identified error is $\widetilde X_1$. The associated recovery is

\begin{equation}
\widetilde R_{000001}
=
\widetilde X_1.
\end{equation}

Therefore,

\begin{align}
\widetilde R_{000001}
\widetilde X_1
\ket{\widetilde\psi_L}
&=
\widetilde X_1^2
\ket{\widetilde\psi_L}
\nonumber\\
&=
\ket{\widetilde\psi_L}.
\label{eq:SM_steane_X1_recovery}
\end{align}

Similarly, the syndrome $001000$ identifies
$\widetilde Z_1$, and the corresponding recovery is

\begin{equation}
\widetilde R_{001000}
=
\widetilde Z_1.
\end{equation}

For the trivial syndrome $000000$, no correction is required, and one
may take

\begin{equation}
\widetilde R_{000000}
=
I_{128}.
\end{equation}

\paragraph{Conclusion of the seven-qubit example.}

The one-hundred-and-twenty-eight-level realization preserves the
complete error-correction structure of the $[[7,1,3]]$ Steane code. In
particular,

\begin{enumerate}
\item the logical codewords are mapped isometrically to
      $\ket{\widetilde0_L}$ and $\ket{\widetilde1_L}$;

\item the transferred generators
      $\widetilde S_1,\ldots,\widetilde S_6$ are commuting Hermitian
      involutions whose simultaneous $+1$ eigenspace is exactly
      $\widetilde{\mathcal C}$;

\item the separation into three $X$-type and three $Z$-type
      stabilizers is preserved, retaining the complete CSS structure;

\item every single-qubit Pauli error is mapped to an explicit level
      permutation, binary-phase operator, or signed level permutation;

\item all twenty-two syndrome assignments associated with the
      identity and the twenty-one single-qubit Pauli errors are
      preserved exactly; and

\item the transferred recovery operator restores the original embedded
      logical state.
\end{enumerate}

The embedded operators therefore reproduce the full logical,
stabilizer, CSS-syndrome, and recovery structure of the seven-qubit
Steane code within a single one-hundred-and-twenty-eight-level system.

This conclusion applies to the transferred error family
$\widetilde{\mathcal E}_7$. It does not by itself imply correction of
every physical operator acting on an arbitrary multilevel extension.
If the physical Hilbert space has dimension $D>128$, errors that couple
the embedded computational sector to additional physical levels
represent genuine leakage processes and must be analyzed separately
using the corresponding physical KL conditions.

\section{S5. Physical-Leakage Syndrome Construction and Recovery}
\label{sec:SM_physical_leakage}

\subsection{Canonical syndrome sectors and recovery}

We now construct the syndrome measurement and recovery operation
associated with the physical-leakage condition established in the main
text. Let
\(\widetilde{\mathcal C}\subseteq\mathcal H_{\mathrm{emb}}\)
be the embedded code space and let \(P\) denote its orthogonal
projector. The physical Hilbert space is decomposed as
\(\mathcal H_D=\mathcal H_{\mathrm{emb}}\oplus
\mathcal H_{\mathrm{leak}}\), where
\(\mathcal H_{\mathrm{leak}}\) contains the physical levels outside the
embedded computational manifold. For each physical error \(F_a\), its
outward-leakage component is
\[
L_a
=
P_{\mathrm{leak}}F_aP_{\mathrm{emb}},
\]
so that \(L_aP\) maps an encoded state from
\(\widetilde{\mathcal C}\) into
\(\mathcal H_{\mathrm{leak}}\).

The main text proves that, once the in-manifold errors are correctable,
the complete physical error family is exactly correctable if and only
if
\[
PL_a^\dagger L_bP
=
\beta_{ab}P
\]
for all \(a,b\). Here, \(\beta=[\beta_{ab}]\) is the
leakage-overlap matrix. Its entries describe the overlaps between the
states produced by different leakage operators on the code space. The
main text also proves that \(\beta\) is Hermitian and positive
semidefinite.

Because \(\beta\) is Hermitian and positive semidefinite, it can be
diagonalized as
\[
U^\dagger\beta U
=
\Lambda,
\qquad
\Lambda_{\mu\nu}
=
\lambda_\mu\delta_{\mu\nu},
\qquad
\lambda_\mu\geq0,
\]
where \(U\) is unitary, \(\Lambda\) is diagonal, and
\(\lambda_\mu\) are the nonnegative eigenvalues of \(\beta\). The
columns of \(U\) define linear combinations of the original leakage
operators,
\[
M_\mu
=
\sum_a U_{a\mu}L_a,
\]
which we call the canonical leakage operators. Since this
transformation is unitary, the families \(\{L_a\}\) and
\(\{M_\mu\}\) span the same physical error space.

Using the leakage condition, the canonical operators satisfy
\[
PM_\mu^\dagger M_\nu P
=
\sum_{a,b}
U_{a\mu}^{*}U_{b\nu}
PL_a^\dagger L_bP
=
\lambda_\mu\delta_{\mu\nu}P.
\]
Thus, different canonical leakage modes have orthogonal action on the
code. If \(\lambda_\mu=0\), then
\(M_\mu P=0\), so that mode does not act on the encoded space and may
be omitted. For every \(\lambda_\mu>0\), define
\[
V_\mu
=
\frac{1}{\sqrt{\lambda_\mu}}M_\mu P.
\]
The operator \(V_\mu\) maps the embedded code into the physical
leakage sector and satisfies
\[
V_\mu^\dagger V_\nu
=
\delta_{\mu\nu}P.
\]
Consequently, each \(V_\mu\) is an isometry on
\(\widetilde{\mathcal C}\), while different values of \(\mu\) have
mutually orthogonal ranges.

The leakage-syndrome sector associated with mode \(\mu\) is therefore
defined by
\[
\mathcal S_\mu
=
V_\mu\widetilde{\mathcal C}.
\]
For arbitrary encoded states
\(\ket{\widetilde\psi_L}\) and
\(\ket{\widetilde\phi_L}\),
\[
\left\langle
V_\mu\widetilde\psi_L
\middle|
V_\nu\widetilde\phi_L
\right\rangle
=
\delta_{\mu\nu}
\braket{\widetilde\psi_L|\widetilde\phi_L}.
\]
Hence the sectors \(\mathcal S_\mu\) are mutually orthogonal, and each
contains an undistorted isometric copy of the original logical state.

The orthogonal projector onto \(\mathcal S_\mu\) is
\[
\Pi_\mu
=
V_\mu V_\mu^\dagger
=
\frac{1}{\lambda_\mu}
M_\mu P M_\mu^\dagger.
\]
Indeed, \(V_\mu^\dagger V_\mu=P\) implies
\(\Pi_\mu^2=\Pi_\mu\), while
\(V_\mu^\dagger V_\nu=\delta_{\mu\nu}P\) gives
\[
\Pi_\mu\Pi_\nu
=
\delta_{\mu\nu}\Pi_\mu.
\]
The projectors \(\Pi_\mu\) therefore distinguish the canonical leakage
syndromes without resolving the logical state stored within each
sector.

Let
\[
Q_{\mathrm{corr}}
=
P+\sum_{\mu:\lambda_\mu>0}\Pi_\mu
\]
be the projector onto the complete correctable subspace, consisting of
the original code space and all correctable leakage-syndrome sectors.
The projector onto the unused remainder of the physical Hilbert space
is
\[
\Pi_{\mathrm{rest}}
=
I_D-Q_{\mathrm{corr}}.
\]
The corresponding projective syndrome measurement is
\[
\mathsf M_{\mathrm{leak}}
=
\left\{
P,\,
\Pi_\mu\;(\lambda_\mu>0),\,
\Pi_{\mathrm{rest}}
\right\}.
\]
The outcome \(P\) indicates that the state lies in the code space, the
outcome \(\Pi_\mu\) identifies leakage into the canonical syndrome
sector \(\mathcal S_\mu\), and
\(\Pi_{\mathrm{rest}}\) collects states outside the designated
correctable subspace.

To see explicitly how the measurement identifies the leakage mode,
note that
\[
M_\nu P
=
\sqrt{\lambda_\nu}V_\nu.
\]
Therefore,
\[
\Pi_\mu M_\nu P
=
V_\mu V_\mu^\dagger
\sqrt{\lambda_\nu}V_\nu
=
\delta_{\mu\nu}M_\nu P.
\]
Thus, a state produced by \(M_\nu\) lies entirely in
\(\mathcal S_\nu\) and produces the syndrome outcome \(\nu\) with
certainty, conditioned on that error branch. Importantly, the
measurement determines only the leakage label and does not distinguish
the logical amplitudes within the state.

After obtaining syndrome \(\mu\), recovery is performed by the inverse
isometry
\[
R_\mu
=
V_\mu^\dagger.
\]
Its action on the canonical leakage modes is
\[
R_\mu M_\nu P
=
\sqrt{\lambda_\nu}\,
\delta_{\mu\nu}P.
\]
For a normalized encoded state
\(\ket{\widetilde\psi_L}\), the unnormalized state produced by
\(M_\mu\) is
\[
M_\mu\ket{\widetilde\psi_L}
=
\sqrt{\lambda_\mu}\,
V_\mu\ket{\widetilde\psi_L}.
\]
After syndrome detection and normalization, the state is
\(V_\mu\ket{\widetilde\psi_L}\). Applying \(R_\mu\) gives
\[
R_\mu V_\mu\ket{\widetilde\psi_L}
=
P\ket{\widetilde\psi_L}
=
\ket{\widetilde\psi_L},
\]
so the logical state, including all amplitudes and coherences, is
restored exactly.

The syndrome measurement and conditional recovery can be combined into
a completely positive trace-preserving map. Let \(\tau\) be any fixed
density operator, chosen for example to be a state in the embedded code
space. Then
\[
\mathcal R_{\mathrm{leak}}(X)
=
PXP
+
\sum_{\mu:\lambda_\mu>0}
V_\mu^\dagger\Pi_\mu X\Pi_\mu V_\mu
+
\operatorname{Tr}
\!\left(
\Pi_{\mathrm{rest}}X
\right)\tau
\]
defines a recovery channel on the full physical Hilbert space. The
first term leaves states already in the code space unchanged, the
second returns each correctable leakage sector to the code through the
corresponding inverse isometry, and the final term assigns an arbitrary
fixed output to the uncorrectable remainder. The latter does not affect
the action of the channel on the designated correctable sectors.

\paragraph*{Example 1: embedded three-qubit repetition code.}

Consider the three-qubit repetition code with logical states
\[
\ket{0_L}=\ket{000},
\qquad
\ket{1_L}=\ket{111}.
\]
Under the binary-to-decimal embedding, these states become
\[
\ket{\widetilde 0_L}=\ket{0},
\qquad
\ket{\widetilde 1_L}=\ket{7},
\]
so that
\[
\widetilde{\mathcal C}
=
\operatorname{span}
\{\ket{0},\ket{7}\},
\qquad
P
=
\ket{0}\bra{0}
+
\ket{7}\bra{7}.
\]

Assume that the physical system contains the additional levels
\(\ket{8},\ldots,\ket{13}\), and consider the outward-leakage
operators
\[
L_1
=
\ket{8}\bra{0}
+
\ket{9}\bra{7},
\qquad
L_2
=
\ket{10}\bra{0}
+
\ket{11}\bra{7},
\qquad
L_3
=
\ket{12}\bra{0}
+
\ket{13}\bra{7}.
\]
Each operator transfers both logical basis states into the same
two-dimensional leakage sector while preserving their relative
amplitudes. For an arbitrary encoded state
\[
\ket{\widetilde\psi_L}
=
c_0\ket{0}
+
c_1\ket{7},
\qquad
|c_0|^2+|c_1|^2=1,
\]
their action is
\[
L_1\ket{\widetilde\psi_L}
=
c_0\ket{8}+c_1\ket{9},
\]
\[
L_2\ket{\widetilde\psi_L}
=
c_0\ket{10}+c_1\ket{11},
\]
and
\[
L_3\ket{\widetilde\psi_L}
=
c_0\ket{12}+c_1\ket{13}.
\]

A direct calculation gives
\[
PL_i^\dagger L_jP
=
\delta_{ij}P.
\]
Thus the leakage-overlap matrix is
\(\beta=I_3\). It is already diagonal, so \(U=I_3\),
\(\lambda_i=1\), and the original leakage operators are themselves the
canonical operators:
\[
M_i=L_i,
\qquad
V_i=L_iP=L_i.
\]

The three orthogonal syndrome sectors are
\[
\mathcal S_1
=
\operatorname{span}\{\ket{8},\ket{9}\},
\qquad
\mathcal S_2
=
\operatorname{span}\{\ket{10},\ket{11}\},
\qquad
\mathcal S_3
=
\operatorname{span}\{\ket{12},\ket{13}\}.
\]
Their projectors are
\[
\Pi_1
=
\ket{8}\bra{8}
+
\ket{9}\bra{9},
\]
\[
\Pi_2
=
\ket{10}\bra{10}
+
\ket{11}\bra{11},
\]
and
\[
\Pi_3
=
\ket{12}\bra{12}
+
\ket{13}\bra{13}.
\]
The syndrome measurement
\[
\left\{
P,\Pi_1,\Pi_2,\Pi_3,\Pi_{\mathrm{rest}}
\right\}
\]
therefore determines which pair of leakage levels is occupied without
measuring the coefficients \(c_0\) and \(c_1\).

The corresponding inverse isometries are
\[
R_1
=
\ket{0}\bra{8}
+
\ket{7}\bra{9},
\]
\[
R_2
=
\ket{0}\bra{10}
+
\ket{7}\bra{11},
\]
and
\[
R_3
=
\ket{0}\bra{12}
+
\ket{7}\bra{13}.
\]
They satisfy
\[
R_iL_jP
=
\delta_{ij}P.
\]
For example, if \(L_2\) acts, the state becomes
\[
L_2\ket{\widetilde\psi_L}
=
c_0\ket{10}
+
c_1\ket{11}.
\]
The syndrome measurement returns outcome \(2\), after which
\[
R_2L_2\ket{\widetilde\psi_L}
=
c_0\ket{0}
+
c_1\ket{7}
=
\ket{\widetilde\psi_L}.
\]
The leakage syndrome is therefore identified and corrected without
revealing or disturbing the encoded logical information.
\paragraph*{Example 2: Binomial Bosonic Code}

To demonstrate that the leakage-correction framework is not restricted
to finite-dimensional stabilizer codes, we consider the single-mode
binomial code \cite{Michael2016}

\[
|0_L\rangle
=
\frac{|0\rangle+|4\rangle}{\sqrt2},
\qquad
|1_L\rangle
=
|2\rangle,
\]

which corrects the error set \(\{I,a\}\), where \(a\) is the bosonic
annihilation operator.

Let

\[
\mathcal C
=
{\rm span}
\{
|0_L\rangle,
|1_L\rangle
\},
\]

with projector

\[
P
=
|0_L\rangle\langle0_L|
+
|1_L\rangle\langle1_L|.
\]

We choose

\[
\mathcal H_{\rm emb}
=
{\rm span}
\{
|0\rangle,|1\rangle,|2\rangle,|3\rangle,|4\rangle
\},
\]

and define

\[
\mathcal H_D
=
\mathcal H_{\rm emb}
\oplus
\mathcal H_{\rm leak}.
\]

Consider the outward leakage operators

\[
L_1
=
|5\rangle\langle0_L|
+
|6\rangle\langle1_L|,
\]

\[
L_2
=
|7\rangle\langle0_L|
+
|8\rangle\langle1_L|.
\]

These satisfy

\[
L_\mu
=
P_{\rm leak}
L_\mu
P_{\rm emb},
\qquad
\mu=1,2,
\]

and therefore represent genuine physical leakage.

For an arbitrary logical state

\[
|\psi_L\rangle
=
c_0|0_L\rangle
+
c_1|1_L\rangle,
\]

the leakage operators produce

\[
L_1|\psi_L\rangle
=
c_0|5\rangle
+
c_1|6\rangle,
\]

\[
L_2|\psi_L\rangle
=
c_0|7\rangle
+
c_1|8\rangle.
\]

A direct calculation gives

\[
PL_\mu^\dagger L_\nu P
=
\delta_{\mu\nu}P,
\]

so that

\[
\beta_{\mu\nu}
=
\delta_{\mu\nu}.
\]

The Leakage Knill--Laflamme condition is therefore satisfied exactly.

The corresponding canonical leakage syndrome sectors are

\[
\mathcal S_1
=
{\rm span}
\{|5\rangle,|6\rangle\},
\]

\[
\mathcal S_2
=
{\rm span}
\{|7\rangle,|8\rangle\}.
\]

The syndrome projectors are

\[
\Pi_1
=
|5\rangle\langle5|
+
|6\rangle\langle6|,
\]

\[
\Pi_2
=
|7\rangle\langle7|
+
|8\rangle\langle8|,
\]

and satisfy

\[
\Pi_\mu\Pi_\nu
=
\delta_{\mu\nu}\Pi_\mu.
\]

The corresponding recovery operators are

\[
R_1
=
|0_L\rangle\langle5|
+
|1_L\rangle\langle6|,
\]

\[
R_2
=
|0_L\rangle\langle7|
+
|1_L\rangle\langle8|.
\]

They obey

\[
R_\mu L_\nu P
=
\delta_{\mu\nu}P,
\]

which implies

\[
R_\mu
L_\mu
|\psi_L\rangle
=
|\psi_L\rangle.
\]

Hence the logical state is recovered perfectly after leakage detection.
This example demonstrates that the Leakage Knill--Laflamme theorem
applies equally to bosonic quantum error-correcting codes and is not
restricted to finite-dimensional stabilizer encodings.

\end{document}